\documentclass[journal]{IEEEtran}
\usepackage[utf8]{inputenc}
\usepackage[autostyle=true]{csquotes} 
\usepackage[english]{babel}
\usepackage{graphicx}
\usepackage{mathtools}
\usepackage{tikz}
\usepackage{hyperref}
\usepackage{xcolor}
\usepackage[linesnumbered,ruled,vlined]{algorithm2e}
\usepackage{amsmath,amssymb}
\usepackage{amsthm}
\usepackage{thmtools,thm-restate}
\usepackage{float}
\newcommand{\remove}[1]{}

\newtheorem{theorem}{Theorem}
\newtheorem{lemma}{Lemma}

\newtheorem{corollary}{Corollary}[theorem]
\theoremstyle{definition}
\newtheorem{definition}{Definition}
\newtheorem{remark}{Remark}

\newcommand{\trunc}{\mbox{\normalfont \tt trunc}}
\newcommand{\sm}{\mbox{\normalfont \tt sum}}

\allowdisplaybreaks
\title{ Bounds and Algorithms for  Alphabetic 
Codes and Binary Search Trees}

\author{R. Bruno, \and  R. De Prisco, \and  A. De Santis,
\and   and 
U. Vaccaro \IEEEmembership{Senior Member,\ IEEE }
\thanks{The authors are  with the Dipartimento di Informatica, 
Universit\`a di Salerno,
		Fisciano (SA), Italy (email: {\tt \{rbruno, robdep, ads, uvaccaro\}@unisa.it}).
  This work was partially supported by project SERICS (PE00000014) under the NRRP MUR program funded by the EU-NGEU.}}

\begin{document}
\maketitle

\thispagestyle{empty}
\pagestyle{plain}
\setcounter{page}{1}
\begin{abstract}
Alphabetic codes and binary search trees are 
combinatorial structures  that abstract
search procedures in  ordered sets endowed with 
probability distributions.
In this paper, we design new linear-time algorithms
to construct alphabetic codes, and we show that 
the obtained codes are not too far
from being optimal.
Moreover, we exploit our results on
alphabetic codes to provide new bounds on the average cost of optimal
binary search trees. Our results improve on the best-known
bounds on  the average cost of optimal
binary search trees present in the literature.
\end{abstract}


\section{Introduction}

The Theory of Variable-Length Codes and Search Theory are 
strongly intertwined. In fact, \textit{any} search procedure
that identifies objects in a given space, by sequentially
executing  appropriate tests naturally defines a variable-length 
encoding of elements in the space. To wit,
one can encode  each  possible test outcome  with 
a different symbol of a
finite alphabet, and the concatenation of
the (encoded) test outcomes, leading to the identification
of an object $x$, gives a \textit{bona fide} encoding of
$x$, for any $x$ in the search space.

More specifically, binary alphabetic codes and binary search trees are basic combinatorial
structures that naturally arise in Search Theory \cite{AW,aigner,ka}. 
In fact, alphabetic codes are \textit{equivalent} to search procedures (in ordered sets) that operate
through comparison queries with binary results,
while binary search trees correspond
to search procedures that use comparison queries whose outcome is ternary.

Alphabetical codes have been extensively studied in the Information Theory
community (see, e.g., \cite{gilbert1959variable, HiYa, 
HORIBE1977148, hu1971optimal, IY, nakatsu1991bounds, sheinwald1992binary, yeung1991alphabetic} and references therein quoted). Equally, binary search trees have been
studied in the Computer Science community, because of the important application scenarios where they appear (see, e.g., \cite{knuth, mehlhorn1975nearly, KM, Nievergelt} and references therein quoted).
The two survey papers \cite{julia, Nagaraj} contain many more references.

It is well known that optimal binary alphabetic codes (i.e., with minimum  average length)
can be constructed in time $O(m\log m)$ \cite{hu1971optimal}, where $m$ is the number of codewords, 
and that optimal binary  search trees  (again, with minimum possible average path length)
can be constructed in time $O(n^2)$ \cite{knuth}, where $n$ is the number of nodes 
in the tree. Most of the research activity 
that followed \cite{hu1971optimal} and  \cite{knuth} has concentrated on the following issues: constructing close-to-optimal binary alphabetic codes (equivalently, close-to-optimal binary search trees)
via \textit{linear-time} algorithms, and providing explicit upper bounds on the cost (average 
path length) of the constructed objects. In this paper, 
we provide novel upper bounds on the cost of 
optimal binary alphabetic codes and optimal binary search trees.
Our bounds improve on the best bounds known in the literature.
Our new results are obtained
by accurately analyzing novel linear-time construction algorithms.


\section{Preliminaries}

Let $S=\{s_1, \ldots, s_m\}$  be a set of symbols,
ordered according to a given total order relation $\prec$, that is, for which it holds that
$s_1\prec \dots \prec s_m$. 
A binary alphabetical code is an order-preserving 
mapping $w:\{s_1, \ldots , s_m\}\mapsto \{0,1\}^+$,
where the order relation on the set of all
binary strings $\{0,1\}^+$ 
is the standard alphabetical order. We denote by $C$ the set of
codewords $C=\{w(s): s\in S\}$, and by 
$\phi=(\phi_1, \ldots , \phi_m)$  the probability distribution
 on the symbols
in $S$, that is, $\phi_i$ is the probability of the $i^{th}$ symbol $s_i$. 
For each $s_i\in S$ we denote by 
  $w(s_i)$  the codeword associated to $s_i$, and by
$\ell(w(s_i))$, or simply by $\ell_i$, the length of $w(s_i)$.

Given an alphabetic code $C$, an order-preserving mapping
$w:s\in \{s_1, \ldots , s_m\}\mapsto w(s)\in \{0,1\}^+$, and 
a probability distribution $\phi=(\phi_1, \ldots , \phi_m)$ on $S$, 
we denote the  average code length of the code  $C$ by 
\begin{equation*}
    \mathbb{E}[C]= \sum_{i=1}^m \phi_i\,\ell_i.
\end{equation*}

In the following sections, we will occasionally see an alphabetic code as a binary tree $T$ where each edge is labeled either by bit 0 or 1, 
and each leaf
of the tree corresponds to a different symbol  $s\in S$. 
Sometimes, it will be convenient to identify an arbitrary leaf $s_i$
of the tree $T$ with
the probability $\phi_i$ of the symbol $s_i\in S$ that is associated to the 
leaf.
The codeword $w(s)$ for the symbol $s\in S$ is equal to
the concatenation of all the 0's and 1's  
in the path from the root of $T$ to the leaf associated with $s$. The order-preserving property
of alphabetic codes implies that the leaves of $T$, read from the leftmost
leaf to
the rightmost one, appear in the order $s_1\prec \dots \prec s_m$.
As said before, alphabetic codes, in their binary tree representation,
naturally correspond to a search procedure that makes use of comparison
queries with binary outcomes.

Binary search trees represent, in a sense, a generalization of binary
trees associated to alphabetic codes, since the former can take into account
both successful and unsuccessful searches. More precisely, 
any comparison-based search algorithm on an ordered list of elements $x_1\prec\dots \prec x_n$ can be represented as a binary search tree, with an internal node 
associated to each element $x_i$ and a leaf associated to each of the $n+1$ gaps $(-\infty, x_1),(x_1,x_2),\dots,(x_{n-1},x_n),(x_n,+\infty)$
among the $x_i$'s. Let $x$ be an element that we want to seek for in the list. 
Each internal node represents a comparison operation between $x$ and 
the content $x_i$ of the node. The outcome of
the comparison can be either $x \prec x_i$, $x=x_i$ or $x \succ x_i$ 
(see Section 6.2 of \cite{knuth}). Thus, the internal nodes correspond to successful searches, and the leaves to the unsuccessful ones.

The cost of a binary search tree $T$ is defined as the average number of questions needed to identify an arbitrary input element $x$. More precisely, 
given a probability distribution $\sigma = (\sigma_1,\dots, \sigma_{2\,n+1}) = (p_0,q_1,p_2,\dots,q_n,p_n)$, where  the probabilities of successful
searches are $q_1,\dots,q_n$, and $p_0,\dots,p_n$ are
the probabilities
of  unsuccessful searches,   the cost of the binary search tree $T$, for the probability distribution $\sigma$, is equal to
\begin{equation}\label{costbst}
\mathbb{E}[T] = \sum_{i=1}^n q_i\,\ell(q_i) + \sum_{i=0}^n p_i\,\ell(p_i),    
\end{equation}
where $\ell(q_i)$ denotes the level of $q_i$ in the tree, i.e., the number of nodes from the root of $T$ to $q_i$, whereas $\ell(p_i)$ is the level of the parent of $p_i$ in the tree \cite{knuth}. {We stipulate that the level of the root of $T$ is equal to 1}.

Note that when the probabilities of successful searches, $q_1,\dots,q_n$, are 
all 0's, the problem of finding an optimal binary search tree for the probability distribution $\sigma$ is equivalent to the problem of finding an optimal alphabetic code for the probability distribution $\phi = (p_0,\dots,p_n)$.


\section{Our contributions in perspective }\label{our}

Alphabetic codes have been extensively
studied. In this section, we limit ourselves to discussing
the work that is most closely related to our findings.
We will present our results 
in Section \ref{alpha} and Section \ref{bst}.
 
In their classic paper \cite{gilbert1959variable}, 
Gilbert and Moore designed a
$O(m^3)$ time  dynamic programming algorithm for the construction of optimal alphabetic codes (that is, of
minimum average length),  where $m$ is the number of symbols. 
Hu and Tucker \cite{hu1971optimal} improved Gilbert and Moore's
result by providing an algorithm of time complexity $O(m\log m)$. Subsequent 
studies have mostly concentrated on 
the problem of giving linear time algorithms for the construction
of almost optimal alphabetic codes and
in providing explicit upper bounds on the average length of 
the constructed codes. It is clear that the derived bounds are
upper bounds on the length of optimal alphabetic codes, as well.

Gilbert and Moore \cite{gilbert1959variable} gave  a linear time algorithm 
to construct an alphabetic code for a set 
of symbols $S$, with associated probability distribution $\phi$, whose average length is less than 
    $H(\phi)+2$. Here, $H(\phi)=-\sum_{i=1}^m\phi_i\log \phi_i$ is the 
    Shannon entropy of the distribution $\phi$. Since the 
    average length of any alphabetic code is lower bounded by $H(\phi)$,
    one gets that 
 Gilbert and Moore's
    codes are at most two bits away from the optimum. We note that Gilbert and Moore's upper bound cannot be improved, unless one has some additional information on the probability distribution $\phi$.
    In fact, one can see that 
    the average length of the best alphabetical code 
    for a set of three symbols $s_1\prec s_2\prec s_3$, with  probability distribution 
     $\phi=(\epsilon, 1-2\epsilon, \epsilon)$, is equal to $2-\epsilon$. On the other hand, 
    $H(\phi)$ can be arbitrarily small.
    Remarkably, the average length of the 
    Huffman code for the same $\phi=(\epsilon, 1-2\epsilon, \epsilon)$
    is equal to $1+2\epsilon$, thus there can be a difference of 
    (approximatively)
    1 bit between  the average length of the best alphabetical code and that of
    the Huffman code, for the same $\phi$.
    The general characterization of the ``shape" of worst-case
    probability distributions for alphabetical codes is given
    in the paper \cite{ks}.
    
    Horibe \cite{HORIBE1977148} provided a better
    upper bound than the  $H(\phi)+2$  bound
    of Gilbert and Moore, by giving an   algorithm to construct 
    alphabetic codes of average length less than 
    \begin{equation}\label{ho}
    H(\phi)+2-(m+2)\phi_{\min}, 
    \end{equation}
    where $\phi_{\min}$ is the 
    smallest probability of $\phi$. 
    
    Nakatsu \cite{nakatsu1991bounds} claimed a different upper bound
    on the minimum average 
    length of alphabetical codes.
    Successively, Sheinwald  \cite{sheinwald1992binary} pointed out 
    a gap in the analysis carried out in \cite{nakatsu1991bounds}.

   Yeung \cite{yeung1991alphabetic} also improved the
   $H(\phi)+2$  bound
    of Gilbert and Moore, by
   proving the following  upper bound on the average length $L_{\min}$ of
    optimal binary alphabetic codes:
\begin{align}
    \label{Yeung_bound1}
   L_{\min} &\leq H(\phi)+2- \phi_1 \left (2-\log \phi_1 -\lceil -\log \phi_1 \rceil\right)
        \nonumber\\
    &\quad- \phi_m \left (2-\log \phi_m -\lceil -\log \phi_m \rceil\right)\\  
          \label{Yeung_bound2}
            &\leq H(\phi)+2-\phi_1 -\phi_m,
\end{align}
where $\phi_1$ and $\phi_m$ are the probabilities of the first and the last symbol of the ordered set of symbols $S$, respectively. 
In general, the bounds (\ref{ho}) and (\ref{Yeung_bound1}) are not
comparable, in the sense that there are 
probability distributions $\phi$ for which (\ref{ho}) is better than 
 (\ref{Yeung_bound1}), and probability distributions $\phi$ for which 
  the opposite holds.

    A new upper bound on the average length of optimal alphabetic codes was recently given by Dagan \textit{ et al.} \cite{dagan}.
    Their bound claims that
    \begin{align} \label{Dagan_bound}
   L_{\min} &\leq H(\phi) +2- \left (\phi_1+\phi_m +\sum_{i=1}^{m-1}\min(\phi_i, \phi_{i+1})\right)\\
            &=H(\phi)+1-\frac{\phi_1+\phi_m}{2}+\frac{1}{2}\sum_{i=1}^{m-1}|\phi_i-\phi_{i+1}|.
\end{align}

While the upper bound (\ref{Dagan_bound}) is correct, in Remark \ref{remark}
(which will be given in Section~\ref{alpha})
we point out
 an unresolved issue in the  analysis given in \cite{dagan} that leads to (\ref{Dagan_bound}).

Our results for alphabetic codes consist of a series of improvements to the upper bound (\ref{Dagan_bound}) of Dagan \textit{ et al.} on the average length 
of optimal alphabetic codes $L_{\min}$. 
For instance, 
we will prove that 
\begin{align}
 L_{\min} &\leq 
 H(\phi)+2 -\phi_1\,(2-\log_2 \phi_1 -\lceil-\log_2 \phi_1\rceil)
 \nonumber\\
 &\quad-\phi_m\,(2-\log_2 \phi_m -\lceil-\log_2 \phi_m\rceil)\nonumber\\ 
 &\quad-\sum_{i=1}^{m-1} \min(\phi_i,\phi_{i+1}).\label{bdv}
\end{align}
More importantly,
we    design a $O(m)$ time algorithm to  construct  alphabetic 
codes of average length at most the right-hand side of (\ref{bdv}). 
The $O(m)$ time algorithm that we design will be instrumental in deriving
the results on binary search trees of Section \ref{bst}.
We derive further improvements to (\ref{Dagan_bound})
in Section \ref{alpha}.

\medskip
In Section \ref{bst} we turn our attention to binary search trees.
Binary search trees have been extensively
investigated (e.g., \cite{de1993binary, mehlhorn1975nearly, KM, Nievergelt}).
Knuth \cite{knuthp} gave an $O(n^2)$ time and $O(n^2)$ space complexity
algorithm to construct an optimum binary
search tree, i.e., a binary search tree with minimum cost, as defined in (\ref{costbst}).
The quadratic complexity of the algorithm designed in
\cite{knuthp} may be prohibitive in many
applications, as Mehlhorn remarked in \cite{mehlhorn1975nearly}. 
Therefore, 
Mehlhorn \cite{mehlhorn1975nearly} proposed a linear time algorithm to construct a binary search tree for an arbitrary probability distribution 
$$\sigma=(\sigma_1, \ldots , \sigma_{2n+1})=(p_0, q_1, p_1, q_2, \ldots ,
q_n,p_n),$$ where the $p_i$'s are probabilities of unsuccessful searches
and the $q_i$'s are probabilities of successful searches,
whose cost is not greater than
\begin{equation}\label{mel}
    H(\sigma) + 1 + \sum_{i=0}^n p_i.
\end{equation}

De Prisco and De Santis \cite{de1993binary} proposed a different linear time algorithm for 
constructing binary search trees  and they claimed  an upper bound
on their cost given by
\begin{equation}
    \label{De_prisco_bound}
    H(\sigma)+1-p_0 -p_n + p_{\max},
\end{equation}
where $p_{\max}$ is the maximum value among $p_0,\dots,p_n$.
Unfortunately, the bound (\ref{De_prisco_bound}) does not hold in general, as pointed out in  \cite{dagan,ku}.

In Section \ref{bst} we correct the analysis 
of the algorithm derived in  \cite{de1993binary}, and we explicitly
estimate the cost of the produced binary search trees.
The simplest of our bounds states
that the cost $\mathbb{E}[T]$ of the tree $T$ constructed by the
algorithm of \cite{de1993binary} is upper bounded by

\begin{align}
 \mathbb{E}[T]
   &\leq H(\sigma) + 1+\sum_{i=0}^n p_i -\Biggl( p_0 +p_n+
    \sum_{i=0}^{n-1}\min(p_i,q_{i+1}) \nonumber\\
    &\quad+\sum_{i=1}^n\min(q_i,p_i)
    + \sum_{i=0}^{n-1} \min(p_i,p_{i+1})\Biggr)\nonumber\\
    &=H(\sigma)+1+\sum_{i=0}^n p_i-\Biggl(1-\sum_{i=0}^{n-1} \frac{|p_i-q_{i+1}|}{2} \nonumber\\
    &\quad- \sum_{i=1}^{n} \frac{|q_i-p_i|}{2} + \sum_{i=0}^n p_i -\sum_{i=0}^{n-1} 
    \frac{|p_i-p_{i+1}|}{2}\Biggr).\label{nostro1}
\end{align}

\medskip
It is apparent that (\ref{nostro1}) significantly improves
the Melhorn bound
(\ref{mel}). We derive additional and better upper bounds than (\ref{nostro1}) 
in Section \ref{bst}.  Clearly, our results also
represent upper limits on the cost of optimal
binary search trees. The new bounds on the
cost of binary search trees are crucially dependent on the new results on alphabetic codes
that we derive in Section \ref{alpha}. 


\section{On Alphabetic codes}\label{alpha}

We first develop some useful machinery.
Given a probability distribution $\phi=(\phi_1,\dots,\phi_m)$, let us denote with
$$\nu= (\nu_1,\dots,\nu_{2\,m-1})=(\phi_1,0,\phi_2,\dots,0,\phi_m)$$ its \textit{partially extended} distribution, and with 
$$\nu\,'=(\nu_1,\dots,\nu_{2\,m+1})=(0,\phi_1,0,\phi_2,\dots,0,\phi_m,0)$$ its \textit{fully extended} distribution.

We also need some auxiliary tools and results from Nakatsu \cite{nakatsu1991bounds}.

\begin{definition}\label{denakatsu}
    \cite{nakatsu1991bounds}
    Let $L = \langle\ell_1,\dots,\ell_m \rangle$ be a list of positive integers.  
    For a binary fraction $x$ and integer $i\geq 1$,
let the function $\trunc$ be defined as
\begin{equation}\label{trunc}
\trunc(i,x)=\frac{\lfloor 2^i\,x\rfloor}{2^i}
\end{equation}
 that is, $\trunc(i,x)$ is the fraction obtained by considering only the first 
 $i$ bits in the binary representation of $x$.
 
 Let $\alpha_i = \min(\ell_{i-1},\ell_i)$, for $i=2,\dots,m$. 
   The recursive  function $\sm$ is  defined as
    \begin{equation}\label{sm}
        \sm(L, i) = \begin{cases}
            \trunc(\alpha_i, \sm(L, i-1))+2^{-\alpha_i} & \text{if } i \geq 2, \\
            0 & \text{if } i = 1.
        \end{cases}
    \end{equation}
\end{definition}

\medskip\noindent

{
To better comprehend the functions \trunc \ and \sm \ defined above, let us 
work out a small numerical example.
Let $L=\{4,2,3,3\}$ be a list of integers. We  have $\alpha_2 = \min(4,2) = 2$, $\alpha_3=\min(2,3) = 2$ and $\alpha_4=\min(3,3) = 3$.
Now, we can compute the value of $\sm$ for $L$. By definition $\sm(L,1)=0$. It holds that
\begin{align*}
    \sm(L,2)&=\trunc(\alpha_2,\sm(L,1)) + 2^{-\alpha_2}\\
    &= \trunc(2, 0) + 2^{-2} = 2^{-2},
\end{align*}
since $\trunc(2, 0)$ takes the value 
whose binary expansion is the same of the first 2 bits of the binary representation of $0$ (that are 
all $0$'s). Similarly, we have
\begin{align*}
     \sm(L,3)&= \trunc(\alpha_3,\sm(L,2))+ 2^{-\alpha_3}\\
     &= \trunc(2, 2^{-2}) + 2^{-2} = 2^{-2} + 2^{-2} = 2^{-1},
\end{align*}
since $\trunc(2, 2^{-2})$ takes the value whose binary expansion is the same of the first 2 bits of the binary representation of $2^{-2}=\frac{1}{4}$,
which correspond to the same value $2^{-2}$ since the binary representation of $\frac{1}{4}$ is $0\,\frac{1}{2} + 1\,\frac{1}{4}$.
Finally, we have
\begin{align*}
     \sm(L,4)&= \trunc(\alpha_4,\sm(L,3))+ 2^{-\alpha_4}\\
     &= \trunc(3, 2^{-1}) + 2^{-3} = 2^{-1} + 2^{-3}.
\end{align*}

Intuitively, when $\sm(L,i)<1$, one can see that $\sm(L,i)$ 
enjoys the property of being the smallest value strictly greater than $\sm(L,i-1)$ that differs from $\sm(L,i-1)$ on 
at least the $\alpha_i^{th}$ 
bit of their binary expansion, where $\alpha_i = \min(\ell_{i-1},\ell_{i})$.  Thus, by denoting with $w_i$ the first $\ell_i$ bits of 
the binary expansion of $\sm(L,i)$ and with $w_{i-1}$ the first $\ell_{i-1}$ bits of 
the binary expansion of 
$\sm(L,i-1)$, the property ensures  that $w_{i-1}$ \textit{comes} before $w_i$ in the
standard alphabetic order among binary sequences.
Moreover, neither  $w_i$ is a  prefix of $w_{i-1}$ 
nor $w_{i-1}$ is a  prefix of $w_{i}$.

\medskip\noindent 
Nakatsu \cite{nakatsu1991bounds} proved the following important result
(an equivalent result was independently given by Yeung \cite{yeung1991alphabetic}).

\begin{theorem}
    \label{condition_1}
    \cite{nakatsu1991bounds}
    Let $L = \langle \ell_1,\dots,\ell_m \rangle$ be 
    a list of integers, associated with the ordered 
    symbols $s_1 \prec \dots \prec s_m$. 
    There exists an alphabetical code for which  the 
    codeword assigned to symbol $s_i$ has length $\ell_i$,
    for each $i=1, \ldots, m$, 
 if and only if $\sm(L, m) < 1$.
\end{theorem}

{
Nakatsu proved the sufficient part of Theorem \ref{condition_1} by showing that the 
first $\ell_i$ digits of the binary representation of $\sm(L, i)$,
for $i=1, \ldots, m$, (under the hypothesis that $\sm(L, m) < 1$)
give the alphabetic code 
for the ordered
set of
symbols $s_1 \prec \dots \prec s_m$. 
}
{In the following Lemma \ref{lemma:linear_time_alg} we show that an alphabetic code, 
whose existence is guaranteed by Theorem \ref{condition_1}, can be constructed 
in time $O(m)$.
We first need the following technical result, whose proof is given in Appendix A.}
{
    \begin{restatable}{lemma}{teclemma}\label{lemma:l_bit}
        Let $L=\langle\ell_1,\dots,\ell_m\rangle$ be a list of integers 
        such that $\sm(L,m)<1$, let 
        $i$, $j$ be integers
        such that $1\leq i<j\leq m$, and let $\{\sm(L,i),\sm(L,i+1),\dots,\sm(L,j)\}$ be a subset of consecutive elements of $\{\sm(L,1),\dots,\sm(L,m)\}$. 
        For each $k=i,\dots,j-1$,
       let   $t_k$ be the \emph{smallest} integer $s$ 
       for which the binary expansion of  
       $\sm(L,k)$ differs from the binary 
       expansion of $\sm(L,k+1)$  on the $s^{th}$ bit, and 
       define $t_{ij}=\min_{i\leq k<j} t_k$. Then, it holds that: 
       \begin{enumerate}
       \item \begin{equation}\label{tij}
            t_{ij} = \lceil-\log_2 (\sm(L,i) \oplus \sm(L,j))\rceil,
        \end{equation}
        where $\sm(L,i) \oplus \sm(L,j)$ is the \emph{numerical} value
        whose binary expansion is the result of  the XOR operation between the binary expansions of $\sm(L,i)$ and $\sm(L,j)$, and 
        \item there exists an index $z\in\{i, \ldots, j-1\}$ such that 
        $\sm(L,z)<\trunc(t_{ij},\sm(L,i))+2^{-t_{ij}}$ \emph{and} $\sm(L,z+1) \geq \trunc(t_{ij},\sm(L,i))+2^{-t_{ij}}.$
       \end{enumerate}
    \end{restatable}
}

{For completeness, we recall that the XOR  $\sm(L,i) \oplus \sm(L,j)$
is easily computable, in most programming languages,  with a single built-in instruction.
Armed with Lemma \ref{lemma:l_bit}, we now prove our constructive result. 
To avoid overburdening the notation, when the pair of integers $(i,j)$
in (\ref{tij}) 
is  clear from the context, we will simply denote 
$t_{ij}$ with $t$.
\begin{lemma}\label{lemma:linear_time_alg}
    Let $L=\langle\ell_1,\dots,\ell_m\rangle$ be a list of integers, associated with the ordered symbols $s_1\prec \dots \prec s_m$. If $\sm(L,m)<1$, then one can construct in $O(m)$ time an alphabetic code $C$ for which the codeword assigned to symbol $s_i$ has length at most $\min(\ell_i, m-1)$, for each $i=1,\dots,m.$
\end{lemma}
\begin{proof}
    To prove the lemma, we design a $O(m)$ time algorithm to 
    construct the {binary} tree associated with the code 
    $C$, that is, the binary tree whose root-to-leaves paths give
    the codewords of $C$. 
    
It is useful to recall that Nakatsu  \cite[Theorem 1]{nakatsu1991bounds} showed that the code constructed by taking the first $\ell_i$ bits of the binary expansion of $\sm(L, i)$ as the codeword associated to the symbol $s_i$, for each $i=1,\dots,m$, is alphabetic.
Our algorithm constructs the binary code tree
    by using  a 
    \textit{subset} of the first $\ell_i$ bits of the binary expansion of $\sm(L, i)$, for each $i=1,\dots,m$. 
    The algorithm builds  the desired tree 
    in a top-down fashion, proceeding by iterated bisections of the set
    \begin{equation}
        {\cal S}= \{\sm(L,1), \sm(L,2), \ldots , \sm(L,m)\}.
    \end{equation}
    Initially, the root of the tree is associated to the 
    whole set
    ${\cal S}$. 
    To perform the first bisection, we proceed as follows.

    For each $a=1,\dots,m-1$,
    let $t_a$ be the {smallest} integer for which the binary expansions of the pair $\sm(L,a)$ and $\sm(L,a+1)$ differ on the $t_a^{th}$ bit, and let $t=t_{1m}=\min_{1\leq a<m}t_a$. From Lemma \ref{lemma:l_bit}, we have that $t=\lceil-\log_2 (\sm(L,1) \oplus \sm(L,m))\rceil$.     
We now compute the  index $k\in \{1, \ldots , m-1\}$ for which it holds that
    \begin{equation}
        \sm(L,k)< \trunc(t,\sm(L,1)) + 2^{-t}\label{primok_1}
    \end{equation}
    and
    \begin{equation}
        \sm(L,k+1)\geq \trunc(t,\sm(L,1))+2^{-t}.\label{primok_2}
    \end{equation}
    We associate the left child of the root to 
    the subset 
    $$\{\sm(L,1), \ldots , \sm(L,k)\},$$ 
    and the 
    right child of the root to the subset 
    $$\{\sm(L,k+1), \ldots , \sm(L,m)\}.$$
    By 2) of Lemma \ref{lemma:l_bit}, both $\{\sm(L,1), \ldots , \sm(L,k)\}$
    and $\{\sm(L,k+1), \ldots , \sm(L,m)\}$ are non empty.
     The subset 
    $\{\sm(L,1), \ldots , \sm(L,k)\}$
    represents the nodes that are in the left subtree of the root. 
    Moreover,  
    by the definition of the function $\trunc$ and from (\ref{primok_1}) and (\ref{primok_2}), each numerical value $\sm(\cdot,\cdot)\in \{\sm(L,1), \ldots , \sm(L,k)\}$ has a  binary expansion with $0$ in the $t^{th}$ position.
    Analogously, the subset  
    $\{\sm(L,k+1), \ldots , \sm(L,m)\}$
    represents the nodes that are in the right subtree of the root and 
    each  value $\sm(\cdot,\cdot)\in \{\sm(L,k+1), \ldots , \sm(L,m)\}$ has a  binary expansion with $1$ in the $t^{th}$ position. 
    The algorithm proceeds recursively on the \textit{left} subset $\{\sm(L,1), \ldots , \sm(L,k)\}$
    and \textit{right} subset $\{\sm(L,k+1), \ldots , \sm(L,m)\}$. It stops when the cardinality of the subset is equal to $1$, in such a case we have just a leaf, 
    or when it is equal to $2$, where we have a final subtree with two sibling leaves.
    The complete pseudo-code of the algorithm is presented in 
    \textbf{Algorithm \ref{alg:alg}}.

    \medskip
    
    \begin{algorithm}
        \label{alg:routine}
        \caption{SubTree($i$,$j$)}
        \If{$i=j$}{
            Construct the leaf $v$ associated with symbol $s_i$.
            
            \Return $v$
        }
        \ElseIf{$i+1=j$}{
            Construct a node $v$ whose left child is the leaf $s_i$ and whose right child is the leaf $s_j$.
            
            \Return $v$
        } 
        \Else{
            $t = \lceil-\log_2 (\sm(L,i)\oplus\sm(L,j))\rceil$
            
            Let $k\in\{i, \ldots, j-1\}$ be the index such that $\sm(L,k)<\trunc(t,\sm(L,i))+2^{-t}$ and $\sm(L,k+1) \geq \trunc(t,\sm(L,i))+2^{-t}$.
            
            $l$ = SubTree($i$,$k$)

            $r$ = SubTree($k+1$,$j$)

            Construct a node $v$ whose left child is the tree $l$ and whose right child is the tree $r$.

            \Return $v$
        }
    \end{algorithm}

    \begin{algorithm}\label{alg:alg}
        \caption{ConstructTree($L$)}
        \KwIn{The list of lengths $L=\langle\ell_1,\dots,\ell_m\rangle$}
        \KwOut{An alphabetic code tree with codeword lengths $\ell'_i\leq 
        \min(\ell_i, m-1)$, for each $i=1,\dots,m.$}
        Compute $\sm(L,i)$ for each $i=1,\dots,n$.

        \Return SubTree(1,$m$)
        
    \end{algorithm}

\bigskip

    The procedure presented above for the construction of the binary tree can be implemented to work in $O(m)$ time and space, by exploiting a technique 
    devised in \cite{fredman1975two}.
    {In fact, the 
    crucial operation in Line 9 of Algorithm \ref{alg:routine} can be implemented in $O(\log \min(s,j-s))$ time, where $s=k-i+1$, by first checking 
    if $\sm(L,r)<\trunc(t,\sm(L,i))+2^{-t}$ with $r=\lceil (j-1+i)/2\rceil$ to discover whether the index $k$ belongs to $\{i,\dots,r-1\}$ or to $\{r,\dots,j-1\}$. Successively,
    we use an exponential search to find the interval in which the value $k$ is located,
    and,  finally, a binary search in such interval to find $k$.
     More precisely, we proceed in the following way: if $k\in\{i,\dots,r-1\}$, we check the validity of  the inequality
    $\sm(L,i+2^v)<\trunc(t,\sm(L,i))+2^{-t}$ in Line 9 for increasing
    values $v=0, 1, 2, 3, \ldots,$
    till we find the first  value of $v$ that satisfies 
    the inequality $\sm(L,i+2^v)\geq \trunc(t,\sm(L,i))+2^{-t}$.
    After $\log s$ step we shall have located the  value of $k$ we seek
    within an interval
    of size at most $2s$.
    A successive binary search will exactly determine the value of $k$.
    All together, this procedure takes $O(\log s)$ steps. Similarly, if $k\in\{r,\dots,j-1\}$, we apply the same procedure with the caveat that we start the exponential search from the right-end of the interval
    $\{r,\dots,j-1\}$. In this way, the procedure will take $O(\log (j-s))$ time. Putting the two cases together, the whole operation takes $O(\log \min(s,j-s))$ time.

    Therefore, by denoting with $T(m)$ the number of operations performed by the algorithm to construct the tree with $m$ leaves, because of Line 9, 10, and 11, we have that
    \begin{align}
        T(m) &\leq\max_{1\leq k\leq m-1}\{T(m-k)+T(k)\nonumber\\&\quad+ \alpha \log_2 \min(k,m-k)+ \beta\} \nonumber\\  
        &= \max_{1\leq k\leq m/2} \{T(m-k)+T(k) + \alpha \log_2 k + \beta\}\label{eq:relazione}, 
    \end{align}
    for suitable constant values
    $\alpha$ and $\beta$.
    }
    
    One can see that the right-hand-side of (\ref{eq:relazione})
   grows at most linearly with $m$
    (e.g., Theorem 11, p. 185 of \cite{KMbook}); 
    therefore  the algorithm {ConstructTree($L$)} requires total time $O(m)$.

    {
    
    The tree produced by algorithm {ConstructTree($L$)} is alphabetic by construction since the recursive bisections
    of the set 
     $$\{\sm(L,1), \sm(L,2), \ldots , \sm(L,m)\}$$
     preserves the symbol order $s_1\prec \dots \prec s_m$.
     What is left to prove is 
     that the procedure constructs a tree
     whose root-to-leaves paths (i.e., codewords)
     have length $\ell'_a\leq\min(\ell_a,m-1)$, for each $a=1,\dots,m$. 

     The inequality $\ell'_a\leq m-1$ holds because the constructed binary
     tree is \textit{full}, that is, each node internal node of the tree has
     two children. This is due to 2) of Lemma \ref{lemma:l_bit} and Line 9,10 and 11 of
     \textbf{Algorithm \ref{alg:routine}}.

    We now prove that $\ell'_a \leq \ell_a$, for each $a=1,\dots,m$. Consider first the cases $a\in\{2,\dots,m-1\}$. By definition of $\sm$, we have that
    \begin{equation}\label{eq:def_sum_a}
        \sm(L,a) = \trunc(\alpha_a, \sm(L,a-1))+ 2^{-\alpha_a},
    \end{equation}
    where $\alpha_a = \min (\ell_{a-1},\ell_a)$.
    Therefore,
    since $\trunc(b,x)$ is the binary fraction one gets by considering only the first $b$ bits of the binary expression of $x$
    (cfr. (\ref{trunc})), from (\ref{eq:def_sum_a}) one obtains that the binary expansion of $\sm(L,a)$ differs from the binary expansion of $\sm(L,a-1)$ on at least the $\alpha_a^{th}$ bit.  
    Because of Line 9,10 and 11 of \textbf{Algorithm \ref{alg:routine}}, this implies that after at most $\alpha_a$ 
    bisection steps the value of $\sm(L,a)$ can no longer be in the same subset of $\sm(L,a-1)$. By symmetry, after at most $\alpha_{a+1}$ steps the value $\sm(L,a)$ can no longer be in the same subset of $\sm(L,a+1)$. All together, after at most $\max(\alpha_a,\alpha_{a+1})$ bisection steps the value $\sm(L,a)$ does not belong to the same subset of either $\sm(L,a-1)$ or $\sm(L,a+1)$, that is, $\sm(L,a)$ 
    must correspond to a leaf in the code-tree. Therefore, 
    the corresponding root-to-leaves path has length 
    \begin{align*}
        \ell'_a &\leq \max(\alpha_a,\alpha_{a+1})
    =\max(\min (\ell_{a-1},\ell_a), \min (\ell_{a},\ell_{a+1}))\\
    &\leq \ell_a.
    \end{align*}
    Analogously, for $a=1$ it holds that
    \begin{equation*}
        \ell'_1 \leq \alpha_2 =\min (\ell_{1},\ell_2)\leq \ell_1,
    \end{equation*}
    and for $a=m$ it holds that 
    \begin{equation*}
        \ell'_m \leq \alpha_m =\min (\ell_{m-1},\ell_m)\leq \ell_m.
    \end{equation*}
    }
\end{proof}

\medskip
    Figure \ref{fig:alphabetic_trees} shows the alphabetic trees produced on the list of lengths $L=\langle 6,6,5,2,9,9,8,6,3,2\rangle$, by taking the first $\ell_i$ bits of the binary expansion of $\sm(L, i)$ (as done by Nakatsu \cite{nakatsu1991bounds})
    and by the application of our  linear time algorithm, respectively.
    As an example, one can see how the codeword associated with the symbol $s_7$ in 
    our code tree, which is $10001$, is only a subsequence of the codeword associated with the symbol $s_7$ in the first tree, that is $\underline{100}00\underline{0}0\underline{1}$, where the underlined bits are those that correspond to the codeword of our
    code  tree.

    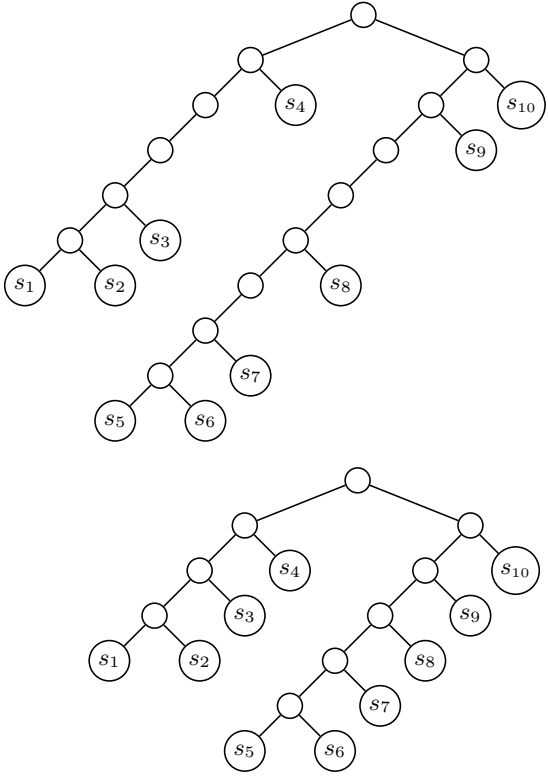
\begin{figure}
        \begin{tabular}{@{}c@{}}
            \begin{tikzpicture}[nodes={circle, draw}, semithick, scale=0.30, level distance=2cm,
            level 1/.style={sibling distance=10cm},
            level 2/.style={sibling distance=4cm},
            level 3/.style={sibling distance=1cm}
            level 4/.style={sibling distance=1cm}]
            \node {}
                child {node {} 
                  child {node {}
                    child {
                            node {}
                            child {
                                node {}
                                child {
                                    node {}
                                    child {
                                        node[inner sep=2pt] {\small$s_1$}
                                    }
                                    child {
                                        node[inner sep=2pt] {\small$s_2$}
                                    }
                                }
                                child { node[inner sep=2pt] {\small$s_3$}}
                            }
                            child[missing] {}
                    }
                    child[missing] {}
                  }
                  child {node[inner sep=2pt] {\small$s_4$}}
                }
                child {node {}
                    child {node {}
                        child {node {}
                                child {
                                    node {}
                                    child {
                                        node {}
                                        child {
                                            node {}
                                            child {
                                                node {}
                                                child {
                                                    node {}
                                                    child {
                                                        node[inner sep=2pt] {\small$s_5$}
                                                    }
                                                    child {
                                                        node[inner sep=2pt] {\small$s_6$}
                                                    }
                                                }
                                                child {node[inner sep=2pt] {\small$s_7$}}
                                            }
                                            child[missing] {}
                                        }
                                        child {node[inner sep=2pt] {\small$s_8$}}
                                    }
                                    child[missing] {}
                                }
                                child[missing] {}
                            }
                        child {node[inner sep=2pt] {\small$s_9$}
                        }
                    }
                    child {node[inner sep=2pt] {\small$s_{\scriptscriptstyle  10}$}   
                    }
                };
            \end{tikzpicture}
             \\[\abovecaptionskip]
        \end{tabular}
        \vspace{\floatsep}
        \hspace{1cm}
        \begin{tabular}{@{}c@{}}
            \begin{tikzpicture}[nodes={circle, draw}, semithick, scale=0.30, level distance=2cm,
            level 1/.style={sibling distance=10cm},
            level 2/.style={sibling distance=4cm},
            level 3/.style={sibling distance=1cm}
            level 4/.style={sibling distance=1cm}]
            \node {}
                child {node {} 
                  child {node {}
                    child {
                            node {}
                            child {
                                node[inner sep=2pt] {{\small$s_1$}}
                            }
                            child {node[inner sep=2pt] {\small$s_2$}}
                    }
                    child {node[inner sep=2pt] {\small$s_3$}}
                  }
                  child {node[inner sep=2pt] {\small$s_4$}}
                }
                child {node {}
                    child {node {}
                        child {node {}
                                child {
                                    node {}
                                    child {
                                        node {}
                                        child {
                                            node[inner sep=2pt] {\small$s_5$}
                                        }
                                        child {node[inner sep=2pt] {\small$s_6$}}
                                    }
                                    child {node[inner sep=2pt] {\small$s_7$}}
                                }
                                child {node[inner sep=2pt] {\small$s_8$}}
                            }
                        child {node[inner sep=2pt] {\small$s_9$}
                        }
                    }
                    child {node[inner sep=2pt] {\small$s_{\scriptscriptstyle10}$}   
                    }
                };
            \end{tikzpicture}
             \\[\abovecaptionskip]
        \end{tabular}
        \caption{The alphabetic trees constructed on the list of lengths $L=\langle 6,6,5,2,9,9,8,6,3,2\rangle$. On the left the tree
        constructed by taking the first $\ell_i$ bits of the binary expansion of
        $\sm(L, i)$, (according to Nakatsu \cite{nakatsu1991bounds}),
         on the right the tree  constructed by our algorithm ConstructTree($L$).}
        \label{fig:alphabetic_trees}
        \end{figure}
}

Before presenting our results on the length of 
the optimal alphabetic code that improves upon the upper bound
(\ref{Dagan_bound}), we briefly show a problematic issue in the analysis of the bound (\ref{Dagan_bound}) claimed in \cite{dagan}.

\begin{remark}
    \label{remark}
    In the proof of the upper bound (\ref{Dagan_bound}) 
    presented in \cite{dagan}, 
    the authors proceed in the following way: Starting from an input probability distribution $\phi=(\phi_1,\dots,\phi_m)$, they first consider the fully extended distribution $\nu\,'=(0,\phi_1,0,\dots,0,\phi_m,0)$
    and, subsequently, apply the classic algorithm of Gilbert and Moore \cite{gilbert1959variable} to construct an alphabetic code $C$ for the distribution $\nu\,'=(0,\phi_1,0,\dots,0,\phi_m,0)$. The average length of the constructed code is less than $H(\nu\,')+2=H(\phi)+2$.
    Successively, the authors prune the binary tree representing the alphabetical code $C$ and show that the average length of the new code so obtained satisfies the upper bound (\ref{Dagan_bound}).
   The  codeword lengths associated with the null
   probabilities are left unspecified in \cite{dagan}, while for 
   non-null probabilities the codeword lengths are equal to 
   $\lceil-\log_2 \phi_i\rceil + 1$, as it is required in the 
   algorithm by Gilbert and Moore.

Next, we uncover the problematic issue in the analysis claimed in \cite{dagan} by a numeric example. Successively, 
in Lemma \ref{lemma_3} of  the Appendix B,
we generalize our example to an infinite number of cases.

\smallskip

  Let $\phi=(\phi_1, \phi_2, \phi_3)=\left(\frac{1}{2},\frac{1}{4}, \frac{1}{4}\right)$ and  $\nu\,'$ be its fully extended distribution $\left(0,\frac{1}{2},0, \frac{1}{4}, 0, \frac{1}{4}, 0\right)$.
    As said above, the length of the codeword produced by the algorithm of Gilbert and Moore for the symbol of probability $\phi_1$ is equal to $\lceil-\log_2 \phi_1\rceil + 1=2$, the length of the codeword produced for the symbol of probability $\phi_2$ is equal to 3, and the length of the codeword produced for the symbol of probability $\phi_3$ is equal to 3.
   However, 
    we will show that there does not exist any alphabetic 
    code for $\nu\,'$ with lengths of the codewords
    \begin{equation}\label{nonesiste}
        L=\langle f_1, \ldots , f_7\rangle=\langle x_1, 2, x_2, 3, x_3, 3, x_4\rangle,
    \end{equation}
    for \textit{any} choice of 
     the codewords length $x_1, x_2, x_3, {x_4}\in \mathbb{N}_+$ associated
     to the symbols of null probability.

    Let us start with the choice $x_1=2$ and $x_2=x_3=x_4=3$. It is immediate to verify that in such a case $\sm(L, 7) = 1$. 
    Therefore, by Theorem \ref{condition_1} there cannot exist an alphabetic code with 
    the chosen lengths. In fact:
    \begin{align*}
        \sm(L,2) &= \trunc(2, \sm(L,1)) + 2^{-2} = 2^{-2}\\
        \sm(L,3) &= \trunc(2, \sm(L,2)) + 2^{-2} = 2^{-1}\\
        \sm(L,4) &= \trunc(3, \sm(L,3)) + 2^{-3} = 2^{-1} + 2^{-3}\\
        \sm(L,5) &= \trunc(3, \sm(L,4)) + 2^{-3} = 2^{-1} + 2^{-2}\\
        \sm(L,6) &= \trunc(3, \sm(L,5)) + 2^{-3} = 2^{-1} + 2^{-2} + 2^{-3}\\
        \sm(L,7) &= \trunc(3, \sm(L,6)) + 2^{-3} = 2^{-1} + 2^{-2} + 2^{-2}\\
        &= 1.
    \end{align*}
    Similarly, if we increase the lengths $x_i$, the value of the function $\sm$ remains the same, since the values of $\alpha_i =\min(f_{i-1}, f_i)$ do not change. If we decrease the lengths $x_i$, considering all possible combinations, we can also verify that the value of $\sm$ remains greater than or equal to 1.
    In fact, let us take for example $x_1=1$, $x_2=x_3=2$ and $x_4=3$, we have that:
    \begin{align*}
        \sm(L,2) &= \trunc(1, \sm(L,1)) + 2^{-1} = 2^{-1}\\
        \sm(L,3) &= \trunc(2, \sm(L,2)) + 2^{-2} = 2^{-1} + 2^{-2}\\
        \sm(L,4) &= \trunc(2, \sm(L,3)) + 2^{-2} = 2^{-1} + 2^{-2}+2^{-2}\\
        &=1\\
        \sm(L,5) &= \trunc(2, \sm(L,4)) + 2^{-2} = {2^{-1} + 2^{-2}+2^{-2}} \\&\quad+ 2^{-2}{=1.25}\\
        \sm(L,6) &= \trunc(2, \sm(L,5)) + 2^{-2} = 2^{-1} + 2^{-2}+2^{-2}\\&\quad+2^{-2} + 2^{-2}{=1.5}\\
        \sm(L,7) &= \trunc(3, \sm(L,6)) + 2^{-3} = 2^{-1} + 2^{-2}+2^{-2}\\&\quad+2^{-2}+2^{-2} + 2^{-3} {= 1.625}.
    \end{align*}
    
    Therefore, by Theorem \ref{condition_1} there cannot exist an alphabetic code with lengths (\ref{nonesiste}).
\end{remark}

\bigskip
We now present our new results for alphabetic codes. We will divide  
our analysis into two cases: The first is about dyadic distributions,
the second case is for everything else.
 We recall that a dyadic distribution $\phi = (\phi_1,\dots,\phi_m)$ is a 
 probability distribution where each $\phi_i$ is equal to $2^{-k_i}$, for suitable integers
$k_i>0$.

\begin{theorem}
    \label{theorem_using_yeung}
    For any dyadic probability distribution $\phi=(\phi_1,\dots,\phi_m)$ 
    on the symbols $s_1 \prec \dots \prec s_m$, we can
    construct in linear time  an alphabetic code $C$ with
    \begin{align}\label{dyadicnostra}  
        \mathbb{E}[C] {\leq} H(\phi)+1-\phi_1-\phi_m. 
    \end{align}
\end{theorem}
\begin{proof}
    The idea of the proof is simple and is based on 
    Yeung's analysis \cite[Thm. 5]{yeung1991alphabetic}. Yeung proved  that, given a probability distribution $\phi = (\phi_1,\dots,\phi_m)$,  there exists 
    an alphabetic code $C$ for $\phi$ with lengths of the codewords $L=\langle\ell_1,\dots,\ell_m\rangle$ equal to:
    \begin{equation*}
    \ell_i = \begin{cases}
            \lceil -\log_2 \phi_i \rceil & \text{ if either } i = 1 
            \text{ or } i= m,\\
            \lceil -\log_2 \phi_i \rceil + 1 & \text{ otherwise}.
        \end{cases}
    \end{equation*}
    {From Lemma \ref{lemma:linear_time_alg}, the code $C$ can be constructed in $O(m)$ time.}

   Using the additional hypothesis that $\phi$ is {a dyadic distribution} we have that $\lceil -\log_2 \phi_i \rceil = -\log_2 \phi_i$ for each $i=1,\ldots,m$. Therefore, the average length of the alphabetic code $C$ with lengths $L$ is
    \begin{align*}
        \mathbb{E}[C] &{\leq} \sum_{i=1}^m \phi_i\,\ell_i\\
        &= \phi_1\,(-\log_2 \phi_1) + \phi_m\,(-\log_2 \phi_m) \\
        &\quad+ \sum_{i=2}^{m-1} \phi_i\,(-\log_2 \phi_i + 1)\\
        &= -\sum_{i=1}^m \phi_i\,\log_2 \phi_i + \sum_{i=2}^{m-1} \phi_i\\
        &= H(\phi) + 1 - \phi_1 -\phi_m.
    \end{align*}
\end{proof}

We now turn our attention to the general case
of non-dyadic probability distributions 
$\phi = (\phi_1,\dots,\phi_m)$.

{Let us first consider the case in which $\phi_1$ and $\phi_m$ are positives.
We also need the following property of the function $\sm$.}

\begin{lemma}
    \label{lemma_2}
    Let $ \ell_1,\dots,\ell_m$ be a sequence of positive
    integers, and let  $k\geq\max_{i} \ell_i$. The function
    $\sm$ on the list 
  $L=\langle f_1, f_2, \ldots , f_{2m-1}\rangle=\langle\ell_1,k,\ell_2,k, \dots,k,\ell_m\rangle$ of $2m-1$ elements satisfies the following
  inequality:
    \begin{equation}
        \sm (L, 2m-1) \leq 2^{-\ell_1} + 2^{-\ell_m} + 2\,\sum_{i=2}^{m-1} 2^{-\ell_i}.
    \end{equation}
\end{lemma}

\begin{proof}
By the definition of number $\alpha_i$ (see Definition \ref{denakatsu}),  it holds that
    $$
        \alpha_2 = \min(f_1,f_2) = \ell_1,
    $$
and 
    $$
        \alpha_{2m-1} = \min(f_{2m-2}, f_{2m-1}) = \ell_{m}.
    $$
For each $i\in\{3, 5, \dots, 2m-3\}$, it holds that
    $$
        \alpha_i = \min(f_{i-1},f_i) = \min(f_i, f_{i+1}) = \alpha_{i+1} = \ell_{\frac{i+1}{2}}.
    $$
By the definition of the function $\sm$,
for each index $i$ we have that
    \begin{align*}
        \sm(L,i) &= \trunc(\alpha_i, \sm(L, i-1))+2^{-\alpha_i}\\
        &\leq \sm(L, i-1) + 2^{-\alpha_i}.
    \end{align*}
Thus, noticing that $\sm(L,1)=0$, we have that
    \begin{align*}
        \sm(L, 2m-1) &\leq \sum_{i=2}^{2m-1} 2^{-\alpha_i}\\ 
        &=  2^{-\ell_1} + 2^{-\ell_m} + 2\,\sum_{i=2}^{m-1} 2^{-\ell_i}.
    \end{align*}
\end{proof}

\begin{theorem}
    \label{theorem_2}
    For any non-dyadic probability distribution $\phi=(\phi_1,\dots,\phi_m)$
    on the symbols $s_1 \prec \dots \prec s_m$ {with $\phi_1,\phi_m>0$}, we can
    construct in linear time  an alphabetic code $C$ with
    \begin{align}
    \mathbb{E}[C] &{\leq} H(\phi)+2-\phi_1\,(2-\log_2 \phi_1 -\lceil-\log_2 \phi_1\rceil)\nonumber\\
    &\quad-\phi_m\,(2-\log_2 \phi_m -\lceil-\log_2 \phi_m\rceil)\nonumber\\&\quad-\sum_{i=1}^{m-1}\min(\phi_i,\phi_{i+1})\label{primobound}\\
    &< H(\phi)+2-\phi_1-\phi_m -\sum_{i=1}^{m-1}\min(\phi_i,\phi_{i+1}).\label{primobound_v2}
    \end{align}
\end{theorem}
\begin{proof}
   We  first show  that 
    there exists an alphabetic code $C'$ for the partially extended distribution 
    $$\nu=(\nu_1, \nu_2, \ldots , \nu_{2m-1})=(\phi_1,0,\dots,0,\phi_m),$$
    with 
    average length upper bounded by 
    \begin{align}
    &H(\phi)+2-\phi_1\,(2-\log_2 \phi_1 -\lceil-\log_2 \phi_1\rceil)\nonumber\\
    &\quad-\phi_m\,(2-\log_2 \phi_m -\lceil-\log_2 \phi_m\rceil).\label{prima}
    \end{align}
  Successively,  we will prune the binary tree representing
$C'$ by removing the additional $m-1$ leaves (corresponding to the 
$m-1$ null probabilities in $\nu$), 
 thus reducing the quantity (\ref{prima}) to (\ref{primobound}).
    
    Let $k$ be an integer greater than $ \max_{j:\phi_j > 0} \lceil-\log_2 \phi_j \rceil +1$ (the value of $k$ will be chosen later, in 
    order to satisfy necessary inequalities), and let $L = \langle \ell_1, \dots, \ell_{2\,m-1}\rangle$ be the lengths of the codewords for  the distribution $\nu=(\phi_1,0,\dots,0,\phi_m)$, 
    defined as follows:
    \begin{equation}
        \ell_i = \begin{cases}\label{defell}
            k & \text{ if } \nu_i=0,\\
            \lceil-\log_2 \nu_i \rceil & \text{ if } i=1 \text{ or } i=2m-1,\\
            \lceil-\log_2 \nu_i \rceil +1 & \text{ if } \nu_i > 0.
        \end{cases}
    \end{equation}
    
    From Lemma \ref{lemma_2} we have that 
    \begin{equation*}
        \begin{aligned}
            \sm(L, 2\,m-1) &\leq 2^{-\lceil -\log_2 \phi_1 \rceil}+ 2^{-\lceil -\log_2 \phi_m \rceil}\\&\quad+ 2\,\sum_{i\neq1,m:\phi_i > 0} 2^{-(\lceil -\log_2 \phi_i \rceil +1)} + 2\,\sum_{i:\phi_i = 0} 2^{-k}\\
            &=\sum_{i:\phi_i > 0} 2^{-\lceil -\log_2 \phi_i \rceil} + 2\,\sum_{i:\phi_i = 0} 2^{-k} < 1,
        \end{aligned}
    \end{equation*}
    where the last inequality holds {because $\sum_{i:\phi_i > 0} 2^{-\lceil -\log_2 \phi_i \rceil}<1$} since $\phi$ is \emph{a non-dyadic} distribution\footnote{Notice that if the distribution is dyadic, this inequality might not hold. However, as we will see later the theorem can be extended to dyadic distributions.}, and we can choose a sufficiently large value for $k$.
    
    By Theorem \ref{condition_1} there exists an alphabetic code with  lengths 
    $\ell_i$ defined in (\ref{defell}),
    and {from Lemma \ref{lemma:linear_time_alg} we can construct a code $C'$ in linear time. The  average length of such a code satisfies
    }

    {
    \begin{align}
        \mathbb{E}[C'] &\leq \phi_1 \lceil-\log_2 \phi_1\rceil +\phi_m \lceil-\log_2 \phi_m\rceil \nonumber\\&
        \quad+ \sum_{i\neq 1,m:\phi_i>0} \phi_i (\lceil-\log_2 \phi_i\rceil +1)\nonumber\\
        &\leq \phi_1 \lceil-\log_2 \phi_1\rceil +\phi_m \lceil-\log_2 \phi_m\rceil \nonumber\\
        &\quad+\sum_{i\neq 1,m:\phi_i>0} \phi_i (-\log_2 \phi_i +2)\nonumber\\&\quad\mbox{(since $\lceil x \rceil < x + 1$)}\nonumber\\
        &= -\phi_1(-\log_2 \phi_1 + 2 -\lceil-\log_2 \phi_1\rceil)\nonumber\\
        &\quad-\phi_m(-\log_2 \phi_m + 2 -\lceil-\log_2 \phi_m\rceil)\nonumber\\
        &\quad+\sum_{i:\phi_i>0} \phi_i (-\log_2 \phi_i +2)\nonumber\\
        &= -\sum_{i:\phi_i>0} \phi_i \log_2 \phi_i + 2\nonumber\\
        &\quad-\phi_1(-\log_2 \phi_1 + 2 -\lceil-\log_2 \phi_1\rceil)\nonumber\\
        &\quad-\phi_m(-\log_2 \phi_m + 2 -\lceil-\log_2 \phi_m\rceil)\nonumber\\
        &= H(\phi)+2-\phi_1(2-\log_2 \phi_1 -\lceil-\log_2 \phi_1\rceil)\nonumber\\
        &\quad-\phi_m(2-\log_2 \phi_m -\lceil-\log_2 \phi_m\rceil).\label{usoyeung} 
    \end{align}
    }
    
    We now eliminate the $m-1$ leaves associated with the {\textit{additional}} \textit{null} probabilities in $\nu$, proceeding
    similarly to \cite{dagan,de1993binary}. 
    First, observe that the leaves
associated with the probabilities $\nu_i$ appear, from left to right,
    in the same order as the probabilities $\nu_i$ (this is due to the fact that the tree represents an alphabetic code). {Moreover, observe that the binary tree of the code $C'$ is full by construction, i.e., each node has exactly zero or two children.}  
    When we prune a leaf $x$ {associated with an additional null probability}, we bring up one 
    level the whole sub-tree $T$ whose root $r$ {is the 
    sibling of $x$.} 
    The tree $T$  contains either the leaf associated 
    with  a probability in $\nu$ that immediately 
    follows the {additional} null probability associated 
    with  $x$, or it contains the leaf associated to a probability in $\nu$ that immediately 
    precedes $x$.  Hence, if $x$ is the ``left" {child} of $a$,  
    its pruning causes a bump up of one level of the leaf that 
    follows $x$, and if $x$ is the 
    ``right" {child} of $a$, it causes a  bump up  of one level of the leaf that 
    precedes  $x$.
    Figure \ref{fig:example} explains the procedure pictorially.

    \begin{figure}[H]
        \centering
        \includegraphics[scale=0.5]{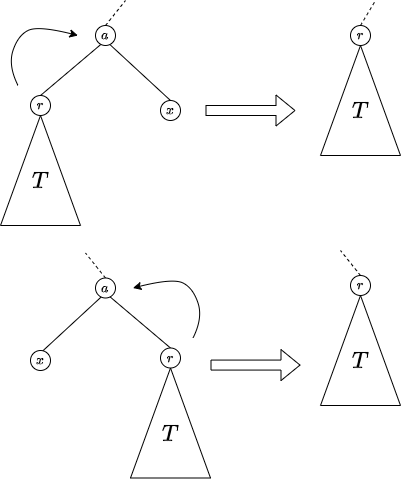}
        \caption{The figure shows the two cases discussed above: the elimination of $x$ bumps up of one level the leaf that precedes it (above), and the leaf that follows it (below).}
        \label{fig:example}
    \end{figure}
    
    Hence, the elimination of each leaf $x$ associated with {an additional} null probability in $\nu$ bumps up one level at least one leaf among the two leaves that precede and follow $x$.
    In total, we bump up leaves whose total probability is at least
    \begin{equation}\label{minfi}
       \sum_{i=1}^{m-1} \min(\phi_i,\phi_{i+1}).
    \end{equation}
    
    We recall
    that the construction of the alphabetic code with average cost (\ref{usoyeung}) requires linear time. Similarly, the elimination of the $m-1$ leaves associated with the {additional} \textit{null} probabilities in $\nu$ requires linear time too, {since it can be implemented through a simple tree visit}. 
    
    Thus, from (\ref{usoyeung}) and (\ref{minfi}) we conclude that 
    the new alphabetic code has a cost upper bounded by
    \begin{align*}
        &H(\phi)+2-\phi_1\,(2-\log_2 \phi_1 -\lceil-\log_2 \phi_1\rceil)\\
        &\quad-\phi_m\,(2-\log_2 \phi_m -\lceil-\log_2 \phi_m\rceil)-\sum_{i=1}^{m-1} \min(\phi_i,\phi_{i+1}).
    \end{align*}
\end{proof}

{Let us now consider the case where $\phi_1$ and $\phi_m$ can be zero.}
{
\begin{theorem}
    For any non-dyadic probability distribution $\phi=(\phi_1,\dots,\phi_m)$
    on the symbols $s_1 \prec \dots \prec s_m$, we can
    construct in linear time  an alphabetic code $C$ with
    \begin{align}
    \mathbb{E}[C] &\leq H(\phi)+2-\phi_{\ell}\,(2-\log_2 \phi_{\ell} -\lceil-\log_2 \phi_{\ell}\rceil)\nonumber\\
    &\quad-\phi_{k}\,(2-\log_2 \phi_{k} -\lceil-\log_2 \phi_{k}\rceil)\nonumber\\
    &\quad-\sum_{i=\ell}^{k-1}\min(\phi_i,\phi_{i+1}) +\phi_{\ell} + \phi_{k}\nonumber\\
    &=H(\phi)+2-\phi_{\ell}\,(1-\log_2 \phi_{\ell} -\lceil-\log_2 \phi_{\ell}\rceil)\nonumber\\
    &\quad-\phi_{k}\,(1-\log_2 \phi_{k} -\lceil-\log_2 \phi_{k}\rceil)-\sum_{i=\ell}^{k-1}\min(\phi_i,\phi_{i+1}),
    \end{align}
    where $\ell = \arg\min_{i} \{\phi_i > 0\}$ and $k = \arg\max_{i} \{\phi_i > 0\}$.
\end{theorem}
\begin{proof}
     Let $\phi'= (\phi_{\ell},\dots,\phi_k)$ be the probability distribution obtained by considering only the elements of $\phi$ between $\phi_{\ell}$ and $\phi_k$. From Theorem \ref{theorem_2}, we know that we can construct in linear time an alphabetic code $C'$ for $\phi'$ with
    \begin{align}
    \mathbb{E}[C'] &\leq H(\phi')+2-\phi_{\ell}\,(2-\log_2 \phi_{\ell} -\lceil-\log_2 \phi_{\ell}\rceil)\nonumber\\
    &\quad-\phi_{k}\,(2-\log_2 \phi_{k} -\lceil-\log_2 \phi_{k}\rceil)-\sum_{i=\ell}^{k-1}\min(\phi_i,\phi_{i+1})\label{eq:codeC'}.
    \end{align}
    We need to address four cases:
    \begin{enumerate}
        \item $\ell=1$ and $k=m$: in this case  $\phi=\phi'$ and the result 
        follows from Theorem \ref{theorem_2}.
        \item $\ell>1$ and $k<m$: then $\phi=(0,\dots,0,\phi_{\ell},\dots,\phi_k,0,\dots,0)$ and we have to create the leaves to allocate 
        the symbols $s_1,\dots,s_{\ell-1}$ and $s_{k+1},\dots,s_n$. To do so, we let the leaves 
        associated to $s_{\ell}$ and $s_k$ descend by one level and 
        we insert
        the subtrees $T$ and $T'$ with leaves $s_1,\dots,s_{\ell-1}$ and $s_{k+1},\dots,s_n$, respectively, as shown in Figure \ref{fig:zeri}. 
        
        \begin{figure}
        \centering
        \includegraphics[scale=0.4]{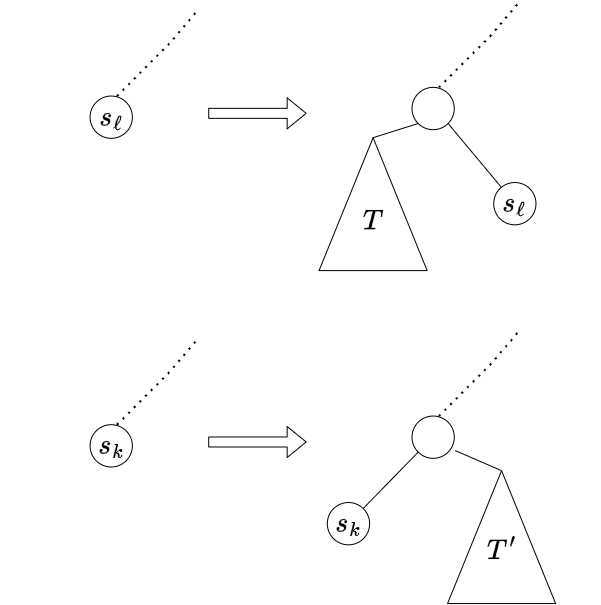}
        \caption{The figure shows the insertion of the subtrees $T$ and $T'$, with leaves $s_1,\dots,s_{\ell-1}$ and $s_{k+1},\dots,s_n$, respectively.}
        \label{fig:zeri}
        \end{figure}

        Thus, from (\ref{eq:codeC'}) and since $H(\phi)=H(\phi')$ we get that the average length of the new code $C$ obtained satisfies
        
        \begin{align}
        \mathbb{E}[C] &\leq H(\phi)+2-\phi_{\ell}\,(2-\log_2 \phi_{\ell} -\lceil-\log_2 \phi_{\ell}\rceil)\nonumber\\
        &\quad-\phi_{k}\,(2-\log_2 \phi_{k} -\lceil-\log_2 \phi_{k}\rceil)\nonumber\\
        &\quad-\sum_{i=\ell}^{k-1}\min(\phi_i,\phi_{i+1}) +\phi_{\ell} + \phi_{k}\nonumber.
    \end{align}
        
        \item $\ell=1$ and $k<m$:
        then $\phi=(\phi_{1},\dots,\phi_k,0,\dots,0)$ and we have to create the leaves for the symbols  $s_{k+1},\dots,s_n$. Therefore, proceeding similarly to case 2) we get that
        \begin{align}
        \mathbb{E}[C] &\leq H(\phi)+2-\phi_{\ell}\,(2-\log_2 \phi_{\ell} -\lceil-\log_2 \phi_{\ell}\rceil)\nonumber\\
        &\quad-\phi_{k}\,(2-\log_2 \phi_{k} -\lceil-\log_2 \phi_{k}\rceil)\nonumber\\&\quad-\sum_{i=\ell}^{k-1}\min(\phi_i,\phi_{i+1}) + \phi_{k}\nonumber\\
        &< H(\phi)+2-\phi_{\ell}\,(2-\log_2 \phi_{\ell} -\lceil-\log_2 \phi_{\ell}\rceil)\nonumber\\
        &\quad-\phi_{k}\,(2-\log_2 \phi_{k} -\lceil-\log_2 \phi_{k}\rceil)\nonumber\\
        &\quad-\sum_{i=\ell}^{k-1}\min(\phi_i,\phi_{i+1}) +\phi_{\ell} + \phi_{k}\nonumber.
    \end{align}
        \item $\ell>1$ and $k=m$: then $\phi=(0,\dots,0,\phi_{\ell},\dots,\phi_m)$ and we have to create the leaves for the symbols $s_1,\dots,s_{\ell-1}$.
        Therefore, again proceeding similarly to case 2) we get a new code $C$ with 
        \begin{align}
        \mathbb{E}[C] &\leq H(\phi)+2-\phi_{\ell}\,(2-\log_2 \phi_{\ell} -\lceil-\log_2 \phi_{\ell}\rceil)\nonumber\\
        &\quad-\phi_{k}\,(2-\log_2 \phi_{k} -\lceil-\log_2 \phi_{k}\rceil)\nonumber\\
        &\quad-\sum_{i=\ell}^{k-1}\min(\phi_i,\phi_{i+1}) + \phi_{\ell}\nonumber\\
        &< H(\phi)+2-\phi_{\ell}\,(2-\log_2 \phi_{\ell} -\lceil-\log_2 \phi_{\ell}\rceil)\nonumber\\
        &\quad-\phi_{k}\,(2-\log_2 \phi_{k} -\lceil-\log_2 \phi_{k}\rceil)\nonumber\\
        &\quad-\sum_{i=\ell}^{k-1}\min(\phi_i,\phi_{i+1}) +\phi_{\ell} + \phi_{k}\nonumber.
        \end{align}
    \end{enumerate}
\end{proof}
}

{
We have proved Theorem \ref{theorem_2} under the hypothesis that the probability
distribution $\phi=(\phi_1,\dots,\phi_m)$ is non-dyadic. One of the referees suggested
the following way to extend Theorem \ref{theorem_2} also to dyadic distributions.

\begin{corollary}\label{cor:ref}
     For any dyadic probability distribution $\phi=(\phi_1,\dots,\phi_m)$ 
    on the symbols $s_1 \prec \dots \prec s_m$, we can
    construct in linear time  an alphabetic code $C$ with
   \begin{align}
    \mathbb{E}[C] &\leq H(\phi)+2-\phi_1\,(2-\log_2 \phi_1 -\lceil-\log_2 \phi_1\rceil)\nonumber\\
    &\quad-\phi_m\,(2-\log_2 \phi_m -\lceil-\log_2 \phi_m\rceil)-\sum_{i=1}^{m-1}\min(\phi_i,\phi_{i+1})\label{bound:ref}\\
    &= H(\phi)+2-2\phi_1 - 2\phi_m -\sum_{i=1}^{m-1}\min(\phi_i,\phi_{i+1}).
    \end{align}
\end{corollary}

\begin{proof}
Let $m\geq 4$. For a given dyadic probability distribution $\phi=(\phi_1,\dots,\phi_m)$, let us define
\begin{equation}\label{eq:hat}
\hat{\phi}=(\hat{\phi}_1,\hat{\phi}_2, \hat{\phi}_3\dots,\hat{\phi}_m)=
(\phi_1,\phi_2-\epsilon, \phi_3+\epsilon, \phi_4,  \dots,\phi_m)
\end{equation}
for a sufficiently small $\epsilon>0$. Since $\hat{\phi}$ is a non-dyadic probability
distribution, we obtain from Theorem \ref{theorem_2}  that  we can
    construct in linear time  an alphabetic code $C$ with
    \begin{align}
    \mathbb{E}[C] &\leq H(\hat{\phi})+2-\hat{\phi_1}\,(2-\log_2 \hat{\phi_1} -\lceil-\log_2 \hat{\phi_1}\rceil)\nonumber\\
    &\quad-\hat{\phi}_m\,(2-\log_2 \hat{\phi}_m -\lceil-\log_2 \hat{\phi}_m\rceil)
    -\sum_{i=1}^{m-1}\min(\hat{\phi}_i,\hat{\phi}_{i+1})\nonumber\\
    &\leq H(\hat{\phi})+2  
   -{\phi_1}\,(2-\log_2 {\phi_1} -\lceil-\log_2 {\phi_1}\rceil)\nonumber\\
   &\quad-{\phi_m}\,(2-\log_2 {\phi_m} -\lceil-\log_2 {\phi_m}\rceil)\nonumber\\
    &\quad-\min(\phi_1, \phi_2-\epsilon) -\min(\phi_2-\epsilon, \phi_3+\epsilon)\nonumber\\
    &\quad-\min(\phi_3+\epsilon, \phi_4)-
    \sum_{i=4}^{m-1}\min({\phi_i},{\phi_{i+1}}).\nonumber
    \end{align}
    Since the Shannon entropy $H(\cdot)$ and 
    the operator $\min(\cdot,\cdot)$ are continuous functions
    for $\epsilon>0$, 
     by letting
    $\epsilon\to 0$ we obtain that
\begin{align*}
    \mathbb{E}[C] &\leq H(\phi)+2-\phi_1\,(2-\log_2 \phi_1 -\lceil-\log_2 \phi_1\rceil)\nonumber\\
    &\quad-\phi_m\,(2-\log_2 \phi_m -\lceil-\log_2 \phi_m\rceil)-\sum_{i=1}^{m-1}\min(\phi_i,\phi_{i+1}).
    \end{align*}
The probability distribution $\hat{\phi}$ is not equal to $\phi$ for $\epsilon>0$. However,
the code tree constructed for $\hat{\phi}$ can be used for $\phi$, provided that
$\epsilon>0$ is sufficiently small.
\end{proof}

We remark that the bounds (\ref{dyadicnostra}) and (\ref{bound:ref}) are not comparable, in the sense that there are dyadic probability distributions $\phi$ for which (\ref{dyadicnostra}) 
is better than (\ref{bound:ref}), and probability distributions $\phi$ for which it holds
the opposite. For instance, if 
a dyadic $\phi=(\phi_1,\dots,\phi_m)$ is ordered in non decreasing fashion, 
then (\ref{bound:ref}) becomes 
$$\mathbb{E}[C] \leq H(\phi)+1-2\phi_1-\phi_m,$$
that is better than the bound (\ref{dyadicnostra}).
On the other hand, if a dyadic $\phi=(\phi_1,\dots,\phi_m)$ is ordered according to the 
``saw-tooth" order, that is, is such that 
$$\phi_1<\phi_2>\phi_3<\phi_4>\ldots >\phi_{m-1}<\phi_m,$$
with $m$ even, then (\ref{bound:ref}) becomes 
\begin{equation}\label{eq:ref}
\mathbb{E}[C] \leq H(\phi)+2\phi_2+2\phi_4\ldots +2\phi_{m-2}-\phi_1.
\end{equation}
One can easily see that there are dyadic probability distributions $\phi=(\phi_1,\dots,\phi_m)$
for which (\ref{dyadicnostra}) is better than (\ref{eq:ref}).
}

\begin{remark}
    It might be worthwhile to point out a way to (possibly) improve the upper
    bound (\ref{primobound})  of Theorem \ref{theorem_2}.  For simplicity,  
    let us assume that all probabilities $\phi_i$ have distinct values (but the reasoning holds in general).
    
    The idea is the following. {Let us consider the lengths defined in (\ref{defell}). After the pruning phase described in Theorem \ref{theorem_2}
    we have that between $\phi_i$ and $\phi_{i+1}$ it is the smaller probability
    to be bumped up, in the worst case. Therefore, let $S_1$ be the set of all indexes $i\in\{2,\dots,m-1\}$ for which holds that $\phi_i = \min(\phi_{i-1},\phi_i)$ \textit{and} $\phi_i=\min(\phi_i,\phi_{i+1})$, i.e., the set of the symbols that in the worst case are bumped up at least two times. Similarly, let $S_2$ be the set of all indexes $i\in \{2,\ldots,m-1\}$ for which holds that $\phi_i = \min(\phi_{i-1},\phi_i)$ or $\phi_i = \min(\phi_{i},\phi_{i+1})$ (but not both), i.e., the set of all symbols that in the worst case are bumped up at least one time (except the first and last symbol). Whereas, for the first and last symbol, let $S_3\subseteq \{1,m\}$ be a set such that $1\in S_3$ if $\min(\phi_1,\phi_2) = \phi_1$, and similarly $m\in S_3$ if $\min (\phi_{m-1},\phi_m)=\phi_m$.  Finally, let $S_4$ be the set of the remaining indexes (expect the first and the last one) that in the worst case are not bumped up, i.e. $S_4 = \{2,\ldots,m-1\}\setminus(S_1\cup S_2)$.  Consequently, 
    we can upper bound  the average length of the obtained alphabetic code in the following way:
    }

    \begin{align}
        \mathbb{E}[C] &\leq
        \sum_{i=2}^{m-1} \phi_i (\lceil-\log_2 \phi_i\rceil+1) + \phi_1(\lceil-\log_2 \phi_1\rceil)\nonumber\\
        &\quad+\phi_m(\lceil-\log_2 \phi_m\rceil)-2\sum_{i\in S_1} \phi_i - \sum_{i \in S_2} \phi_i
        -\sum_{i\in S_3} \phi_i\nonumber\\
        &=\sum_{i\in S_4} \phi_i(\lceil-\log_2 \phi_i\rceil+1)\nonumber\\
        &\quad+ \sum_{i\in S_1 \cup S_2 \cup \{1,m\}} \phi_i (\lceil-\log_2 \phi_i \rceil)-\sum_{i\in S_1} \phi_i -\sum_{i\in S_3} \phi_i\label{eq:avg_length}\\
        &<\sum_{i\in S_4} \phi_i(-\log_2 \phi_i+2) + \sum_{i\in S_1 \cup S_2 \cup \{1,m\}} \phi_i (\lceil-\log_2 \phi_i \rceil)\nonumber\\
        &\quad-\sum_{i\in S_1} \phi_i -\sum_{i\in S_3} \phi_i\nonumber\hspace{1cm} \mbox{(since $\lceil x \rceil < x+1$)}
        \\
        &= \sum_{i\in S_4} \phi_i(-\log_2 \phi_i+2) + \sum_{i\in S_1 \cup S_2 \cup \{1,m\}} \phi_i (\lceil-\log_2 \phi_i \rceil)\nonumber\\
        &\quad-\sum_{i\in S_1} \phi_i -\sum_{i\in S_3} \phi_i-\sum_{i\in S_1 \cup S_2 \cup \{1,m\}} \phi_i (-\log_2 \phi_i +2)\nonumber\\
        &\quad+\sum_{i\in S_1 \cup S_2 \cup \{1,m\}} \phi_i (-\log_2 \phi_i +2)\nonumber\\
        &= \sum_{i=1}^m \phi_i(-\log_2 \phi_i+2)\nonumber\\
        &\quad- \sum_{i\in S_1 \cup S_2 \cup \{1,m\}} \phi_i (2 -\log_2 \phi_i -\lceil-\log_2 \phi_i \rceil)\nonumber\\
        &\quad- \sum_{i\in S_1} \phi_i -\sum_{i\in S_3} \phi_i\nonumber\\
        &=H(\phi) + 2 - \sum_{i\in S_1 \cup S_2 \cup \{1,m\}} \phi_i (2 -\log_2 \phi_i -\lceil-\log_2 \phi_i \rceil) \nonumber\\
        &\quad- \sum_{i\in S_1} \phi_i -\sum_{i\in S_3} \phi_i\label{eq:avg_length2}\\
        &< H(\phi) + 2 - \phi_1 (2 -\log_2 \phi_1 -\lceil-\log_2 \phi_1 \rceil)\nonumber\\
        &\quad-\phi_m (2 -\log_2 \phi_m -\lceil-\log_2 \phi_m \rceil)-\sum_{i\in S_1 \cup S_2} \phi_i -\sum_{i\in S_1 \cup S_3} \phi_i\nonumber\\
        &\quad\mbox{(since $2-\log_2 \phi_i -\lceil-\log_2 \phi_i \rceil > 1$)}\nonumber\\ 
        &= H(\phi) +2 -\phi_1\,(2-\log_2 \phi_1 -\lceil-\log_2 \phi_1 \rceil)\nonumber\\
        &\quad-\phi_m\,(2-\log_2 \phi_m -\lceil-\log_2 \phi_m \rceil)-\sum_{i=1}^{m-1} \min(\phi_i,\phi_{i+1})\nonumber\\
        &\quad\mbox{(since $\sum_{i\in S_1 \cup S_2} \phi_i +\sum_{i\in S_1 \cup S_3} \phi_i = \sum_{i=1}^{m-1} \min(\phi_i,\phi_{i+1})).$}\nonumber
    \end{align}

    Therefore,   (\ref{eq:avg_length2}) and  (\ref{eq:avg_length}) 
     \textit{cannot be worse}  than the upper bound (\ref{primobound}) of Theorem \ref{theorem_2}. 
     We notice that the  right-hand side  of
     (\ref{eq:avg_length})  and (\ref{eq:avg_length2}) depends on 
    the   content of the sets $S_1, 
    S_2$, $S_3$ and {$S_4$}. Working out an explicit estimate 
    of the sums appearing on  the right-hand side of  (\ref{eq:avg_length})  and (\ref{eq:avg_length2}) 
    would provide further improvements to
    (\ref{primobound}).
    However, so doing seems challenging. We leave it as an open problem for future investigations.
\end{remark}

\smallskip

Theorem \ref{theorem_2} can be explicitly improved (under a slightly stronger
hypothesis), as shown in the following Corollary.
\begin{corollary}
    \label{corollary_2}
    For any  
    probability distribution $\phi=(\phi_1,\dots,\phi_m)$ on the symbols $s_1 \prec \dots \prec s_m$, with ${\phi_1},\phi_2, \phi_{m-1},{\phi_m}>0$, we can
    construct in linear time  an alphabetic code $C$ with
    \begin{equation}
        \begin{aligned}
            \mathbb{E}[C] &{\leq} H(\phi)+2-\phi_1(2-\log_2 \phi_1 -\lceil-\log_2 \phi_1\rceil)\nonumber\\
            &\quad-\phi_m(2-\log_2 \phi_{m} -\lceil-\log_2 \phi_{m}\rceil)\nonumber\\
            &\quad-\sum_{i=1}^{m-1} \min(\phi_i,\phi_{i+1})\nonumber\\
            &\quad-\phi_1\max(0,\lceil-\log_2 \phi_1 \rceil-\lceil-\log_2 \phi_2\rceil-2)\nonumber\\
            &\quad-\phi_m\max(0,\lceil-\log_2 \phi_m \rceil-\lceil-\log_2 \phi_{m-1}\rceil-2)\label{dacorollary_2}.
        \end{aligned}
    \end{equation}
\end{corollary}
\begin{proof}
 {Assume first that  $\phi$ is non dyadic.}
   We {consider} the same code $C$ (equivalently, tree)  constructed in Theorem \ref{theorem_2},
   and we take advantage of the knowledge of the lengths of the codewords of the first two symbols and of the last two symbols. In fact, such knowledge in these two particular cases allows us to state more precisely how many levels we bump up the first and last symbol during the {construction of the code $C$}. 
    
    Let  $\ell_1=\lceil -\log_2 \phi_1 \rceil$ and $\ell_2 = \lceil -\log_2 \phi_2 \rceil +1$ be the lengths of the codewords associated to symbols $s_1$ and $s_2$, 
    according to (\ref{defell}). 
    When the level of the codeword associated to $s_1$ exceeds the level of the codeword of $s_2$ by \textit{more} than one
    unit, {that is when $\ell_1-\ell_2\ge 2$}, 
    it holds that 
    for \textit{each} of these additional levels, there {would have been}  an internal node having only one child {in the tree representation of the code}. {However, these nodes are removed since the intermediate code $C'$ constructed through Lemma \ref{lemma:linear_time_alg} is full in its tree representation.} The elimination of such nodes bumps up the symbol $s_1$, since the nodes {would have been} on the path from the lowest common ancestor of $s_1$ and $s_2$ to $s_1$. Thus, when {$\ell_1-\ell_2-1>0$} we can bump up $\phi_1$ of at least {$\ell_1-\ell_2-1$ additional} levels.
    {This gives the $-\phi_1\max(0,\lceil-\log_2 \phi_1 \rceil-\lceil-\log_2 \phi_2\rceil-2)$ term of the bound.}
   
    A similar reasoning can be made with $s_m$ and $s_{m-1}$,
    where  $\ell_m =\lceil -\log_2 \phi_m\rceil$ and $\ell_{m-1} = \lceil -\log_2 \phi_{m-1} \rceil +1$.
    Note that such a reasoning cannot be applied to the other pairs of probabilities, since the 
    knowledge of the codeword lengths alone
    is not sufficient to understand which of the symbols will be bumped up. We have to look at the whole codewords set
    to have enough information.

    {The bound (\ref{dacorollary_2}) can be extended to dyadic probability distributions proceeding
    as in Corollary \ref{cor:ref}.}
\end{proof}


\section{On Binary search trees}\label{bst}

In this section, we exploit the results on alphabetic codes
we have derived in Section \ref{alpha} to provide new bounds
on the average cost of optimal binary search trees. We recall that
Knuth's algorithm \cite{knuthp} to construct an optimum binary
search tree has quadratic time and space complexity; here we provide
a linear-time algorithm and give upper bounds on the cost of
the trees our algorithm produces. It is clear that our 
upper bounds apply to optimal trees as well;  we will also 
see that our bounds improve on the best previous results.
Since  our algorithms to construct almost-optimal search trees
take in input almost-optimal alphabetic trees, in order
to get linear time complexity it is essential  that we can
construct almost-optimal alphabetic trees in linear time as well.
This is one of the main reasons why we have stressed 
this feature of the algorithms presented in Section \ref{alpha}.

\medskip
We start this section  by amending the analysis of
the algorithm  presented by De Prisco and De Santis \cite{de1993binary}
(actually, our presentation  is slightly more
general than that of \cite{de1993binary}).

\begin{theorem}
    \label{theorem_6}
    Let $\sigma=(\sigma_1,\dots,\sigma_{2\,n+1})=(p_0,q_1,p_1,\dots,q_n,p_n)$ be a probability distribution that  contains
    the probabilities $q_1,\dots,q_n$ and $p_0,\dots,p_n$ 
    for successful and unsuccessful searches, respectively, 
    on an ordered sequence $x_1 \prec \dots \prec x_n$.
    Let $C$ be \emph{any} alphabetic code 
    for $\sigma$, of average length  $\mathbb{E}[C]$.
    Then, one can construct in linear time
    a binary search tree $T$ for $\sigma$  of cost
    \begin{equation}\label{general}
            \mathbb{E}[T] \leq \mathbb{E}[C] -\sum_{i=1}^n q_i -\sum_{i=0}^{n-1} \min(p_i,p_{i+1}). 
    \end{equation}
\end{theorem}
\begin{proof}
  Let us start from the binary tree $B$ corresponding to 
  the alphabetic code $C$. 
    The tree $B$ has $2n+1$ leaves, that we identify with the ordered
    sequence of probabilities
    $p_0,q_1,p_1,\dots,q_n,p_n$. We transform the 
    leaves corresponding to the $q_i$'s into \textit{internal} nodes. 
    As in \cite{de1993binary}, we raise each $q_i$ to the lowest common ancestor of $p_{i-1}$ and $p_i$, which is at least two levels above $q_i$. Hence, $\ell(q_i)$, i.e. the number of nodes from the root to $q_i$, decreases by at least of one unit, for each $i$. In addition, when we bring up 
    $q_i$ the whole sub-tree whose root is the sibling of $q_i$ goes up by at least one level. Hence, at least one between  $p_{i-1}$ and $p_i$ goes up one level. \footnote{The analysis of \cite{de1993binary} is incorrect since it claims that all but only one probability goes up one level. One can
    see that this is not the case, as also observed by the authors of \cite{dagan} and by \cite{ku}.} Figure  \ref{fig:algorithm_example} provides an example of the described process. 
    
    \begin{figure}[H]
        \begin{tabular}{@{}c@{}}
            \begin{tikzpicture}[nodes={circle, draw}, semithick, scale=0.4, level distance=2cm,
            level 1/.style={sibling distance=10cm},
            level 2/.style={sibling distance=4cm},
            level 3/.style={sibling distance=1cm}
            level 4/.style={sibling distance=1cm}]
            \node {}
                child {node {} 
                  child {node {}
                    child {node[rectangle, draw] {$p_0$}}
                    child {node {$q_1$}}
                  }
                  child {node[rectangle, draw] {$p_1$}}
                }
                child {node {}
                    child {node {}
                        child {node {$q_2$}}
                        child {node {}
                            child {node[rectangle, draw] {$p_2$}}
                            child {node {$q_3$}}
                        }
                    }
                    child {node[rectangle, draw] {$p_3$}   
                    }
                };
            \end{tikzpicture}
             \\[\abovecaptionskip]
            \small (a) The initial alphabetic tree
        \vspace{0.2cm}
        \end{tabular}
        
         \begin{tabular}{@{}c@{}}
            \hspace{0.3cm}
            \begin{tikzpicture}[nodes={circle, draw}, semithick, scale=0.40, level distance=2cm,
                level 1/.style={sibling distance=10cm},
                level 2/.style={sibling distance=4cm},
                level 3/.style={sibling distance=1cm}
                level 4/.style={sibling distance=1cm}]
                \node {}
                child {node {$q_1$} 
                  child {node[rectangle, draw] {$p_0$}}
                  child {node[rectangle, draw] {$p_1$}}
                }
                child {node {}
                    child {node {}
                        child {node {$q_2$}}
                        child {node {}
                            child {node[rectangle, draw] {$p_2$}}
                            child {node {$q_3$}}
                        }
                    }
                    child {node[rectangle, draw] {$p_3$}   
                    }
                };
            \end{tikzpicture} \\[\abovecaptionskip]
            \small (b) The tree after transforming $q_1$ into an internal node
            \vspace{0.2cm}
        \end{tabular}
        \begin{tabular}{@{}c@{}}
            \hspace{0.55cm}
            \begin{tikzpicture}[nodes={circle, draw}, semithick, scale=0.40, level distance=2cm,
            level 1/.style={sibling distance=10cm},
            level 2/.style={sibling distance=4cm},
            level 3/.style={sibling distance=1cm}
            level 4/.style={sibling distance=1cm}]
            \node {$q_2$}
                child {node {$q_1$} 
                  child {node[rectangle, draw] {$p_0$}}
                  child {node[rectangle, draw] {$p_1$}}
                }
                child {node {$q_3$}
                    child {node[rectangle, draw] {$p_2$}}
                    child {node[rectangle, draw] {$p_3$}   
                    }
                };
            \end{tikzpicture} \\[\abovecaptionskip]
            \small (c) The final binary search tree
        \end{tabular}
        \caption{Figure (a) shows the initial alphabetic tree; Figure (b) shows the tree after performing the first step; Figure (c) shows the tree after the application of the algorithm.}
        \label{fig:algorithm_example}
    \end{figure}

    Therefore, in total we bump up leaves whose total probability is at least
    \begin{equation}\label{bump_up_probabilities}
        \sum_{i=1}^n q_i + \sum_{i=0}^{n-1} \min(p_i,p_{i+1}).
    \end{equation}

    Let us show that the   tree $T$, constructed according 
    to the above procedure, is a correct binary search tree for the distribution  $\sigma=(\sigma_1,\dots,\sigma_{2\,n+1})=(p_0,q_1,p_1,\dots,q_n,p_n)$. 
    
    To that purpose, let us observe that the search property holds
     in a binary tree if and only if its in-order visit 
     gives the nodes in the order $p_0, q_1, p_1,\ldots,q_n,p_n$. To show that such a property holds in our tree $T$, we first point out
    that each probability $q_i$ necessarily occupies a distinct internal node in $T$. By 
    the sake of contradiction let us suppose that it is not the case,
    that is, there exist $q_i$ and $q_j$, with $i<j$, that have been assigned to the same node $x$. Then, $x$ is the lowest common ancestor  both of $p_{i-1}$ and $p_{i}$, and of $p_{j-1}$ and $p_j$. This 
    implies that $p_{i-1}$ and $p_{j-1}$ are in the left sub-tree of $x$, while $p_i$ and $p_j$ are in the right sub-tree of $x$. 
    This flagrantly contradicts the fact that  in the tree $B$, from which we started,
    the leaves appear in the order $p_0,q_1,p_1,\dots,q_n,p_n$ (i.e., that $B$   
    represents an alphabetic code).
    
    Thus, each $q_i$ is assigned to a distinct lowest common ancestor, and the order of the $p_i$'s is not affected by the transformation of the $q_i$'s in internal nodes. Therefore, an in-order visit of the tree gives the nodes in the order $p_0, q_1, p_1,\ldots,q_n,p_n$. We conclude that $T$ is a correct binary search tree for $\sigma$.

    The manipulation of  the $q_i$'s to transform them
     from leaves of $B$ to internal nodes of $T$ can be implemented in linear time, as shown in \cite{de1993binary}.
    
   Finally, we evaluate the cost of the obtained 
    binary search tree $T$. From (\ref{bump_up_probabilities}) the cost of  $T$ satisfies the inequality
    \begin{equation*}
         \mathbb{E}[T] \leq \mathbb{E}[C] -\sum_{i=1}^n q_i -\sum_{i=0}^{n-1} \min(p_i,p_{i+1}).
    \end{equation*}
    This concludes the proof of the theorem.
\end{proof}

\medskip
By plugging into (\ref{general}) either formula (\ref{dyadicnostra})
of Theorem \ref{theorem_using_yeung}, or formula (\ref{primobound})
of Theorem \ref{theorem_2}, or formula (\ref{dacorollary_2}) of
Corollary \ref{corollary_2} (according to the hypothesis that holds)
we get a series of upper bounds on the
cost of optimal binary search trees that considerably improve on the bound
(\ref{mel}) of \cite{mehlhorn1975nearly}, as we asserted in Section \ref{our}.
For instance, by using formula (\ref{primobound_v2})
of Theorem \ref{theorem_2} in formula (\ref{general}), we get 
the  upper bound 
\begin{align}
    &H(\sigma) + 2 -\sigma_1 -\sigma_{2n+1} -\sum_{i=1}^{2n} \min(\sigma_i, \sigma_{i+1})\nonumber\\ 
    &\quad-\sum_{i=1}^n q_i - \sum_{i=0}^{n-1} \min(p_i,p_{i+1}),\label{bis}
\end{align}
already  
mentioned in Section \ref{our}.

We remark that our upper bounds on the cost of optimal  binary search trees are,
in fact, upper bounds on the cost of binary search trees constructed by
the \textit{linear} time algorithm outlined in Theorem \ref{theorem_6}.

\medskip
It is worth noting that the binary search trees constructed
according to Theorem \ref{theorem_6} are not too far from being
optimal, and this is due to the lower bounds on the cost of \textit{any} binary search tree given in \cite{de1996binary}. 
For example,
the authors of \cite{de1996binary} proved that 
the cost of {any} binary search tree is lower bounded by
\begin{equation}\label{lower}
H(\sigma)+ \sum_{i=1}^n q_i + H(\sigma)\log_2(H(\sigma)) -
(H(\sigma)+1)\log_2(H(\sigma)+1).
\end{equation}
One can  see that (\ref{bis}) and (\ref{lower}) are not too far apart.
{In fact, the difference between the upper bound (\ref{bis}) and the lower bound
(\ref{lower}) is smaller than
\begin{align}
2& -\sigma_1 -\sigma_{2n+1} -\sum_{i=1}^{2n} \min(\sigma_i, \sigma_{i+1}) -2\sum_{i=1}^n q_i\nonumber\\
\quad&-\sum_{i=0}^{n-1} \min(p_i,p_{i+1})\nonumber \\
\quad& -H(\sigma)\log_2(H(\sigma)) +
(H(\sigma)+1)\log_2(H(\sigma)+1)\nonumber \\
&< 2 +\log_2(H(\sigma)+1) +\log_2e -\sigma_1 -\sigma_{2n+1}\nonumber\\
&\quad-\sum_{i=1}^{2n} \min(\sigma_i, \sigma_{i+1}) -2\sum_{i=1}^n q_i - \sum_{i=0}^{n-1} \min(p_i,p_{i+1}).\label{diffupplow}
\end{align}

Additionally, other results in \cite{de1996binary} strengthen the lower bound (\ref{lower}) for particular
range of values of the entropy $H$ and thus can be used the
reduce the gap between the upper bound and lower bound
given in (\ref{diffupplow}).}
Therefore, in situations where the 
optimal but onerous $O(n^2)$ time algorithm  by Knuth  \cite{knuthp} cannot be used, 
our analysis shows that the  $O(n)$ time  algorithm of \cite{de1993binary}
is a valid alternative. 
An additional merit of Theorem \ref{theorem_6} is that any possibly improved bounds
on the average length of alphabetic codes will immediately translate
into improved bounds on the cost of binary search trees.

\section*{Acknowledgments}
The authors would like to thank the referees 
for their comments, corrections,  and insightful suggestions. 
Our paper has been greatly improved by the
generous help of the referees.



\section*{Appendix A: Proof of Lemma \ref{lemma:l_bit}}
\label{appendix:a}

{

\begin{proof}
To prove claim 1) of the lemma, 
let $s$ be the  \textit{smallest} integer for which the binary expansions of the pair $\sm(L,i)$ and $\sm(L,j)$ differ on the $s^{th}$ bit. We 
first prove that $s$ and 
$t=t_{ij}$ (cfr, eq. (\ref{tij})) have the same value. 

First of all, by definition of $s$, we know that the binary expansions of $\sm(L, i)$ and $\sm(L,j)$ are equal on the first $s-1$ bits. Moreover, since $\sm(L,i)$ and $\sm(L,j)$ are equal on the first $s-1$ bits and $\sm(L,i)<\dots<\sm(L,j)$, it follows that for each $k=i,\dots,j-1$, the binary expansions of $\sm(L,k)$ and $\sm(L,k+1)$ are equal on the first $s-1$ bits, too. Therefore, it holds that $s\leq t$. 

Since $\sm(L,i)$ and $\sm(L,j)$ differ on the $s^{th}$ bit and $\sm(L,i)<\sm(L,j)$, we know that the $s^{th}$ bit of $\sm(L,i)$ is 0, while the $s^{th}$ bit of $\sm(L,j)$ is 1. As a consequence, since $\sm(L,i)<\dots<\sm(L,j)$ and for each $k=i,\dots,j-1$, 
\textit{both} the binary expansions of $\sm(L,k)$ \textit{and} of
$\sm(L,k+1)$ are equal on the first $s-1$ bits, there must exist an index $k\in\{i,\dots,j-1\}$ such that the $s^{th}$ bit of $\sm(L,k)$ is 0, and the $s^{th}$ bit of $\sm(L,k+1)$ is 1. 
Therefore, since $t=t_{ij}=\min_{i\leq k<j} t_k$ and there exists $k$ such that $t_k=s$, it holds that $s\geq t$. Thus, we finally get that $s=t$. 

One can also see that $s=\lceil-\log_2 (\sm(L,i) \oplus \sm(L,j))\rceil$, where
$\sm(L,i) \oplus \sm(L,j)$ is the numerical value
        whose binary expansion is equal to the XOR operation between the binary expansions of $\sm(L,i)$ and $\sm(L,j)$.
        Indeed, since
the binary expansions of $\sm(L, i)$ and $\sm(L,j)$ 
are equal on the first $s-1$ bits, it holds that the 
binary expansion of the value $\sm(L,i) \oplus \sm(L,j)$ contains the first bit 
equal to 1 exactly in the $s^{th}$ position. 
Therefore, one gets that $s$ is the smallest integer for which it holds that
    \begin{equation}\label{eq:ell_def}
        2^{-s}\leq \sm(L,i)\oplus\sm(L,j).
    \end{equation}
    From (\ref{eq:ell_def}), we get that $s=\lceil-\log_2 (\sm(L,i) \oplus \sm(L,j))\rceil$.   

{
To prove claim 2) of the lemma, 
let $z\in\{i,\dots,j-1\}$ denote the index for which it holds that $t_z=t_{ij}$. Then, we have that the binary expansion of $\sm(L,z)$ differs from the binary expansion of $\sm(L,z+1)$ on the $t_{ij}^{th}$ bit. In particular, from the definition of $t_z$, the binary expansions of  $\sm(L,z)$ and $\sm(L,z+1)$ are equal on the first $t_{ij}-1$ bits, and, since $\sm(L,z)<\sm(L,z+1)$, we have that $\sm(L,z)$ has a binary expansion with $0$ in the $t_{ij}^{th}$ position, while  $\sm(L,z+1)$ has a binary expansion with $1$ in the $t_{ij}^{th}$ position. Finally, since $\sm(L,i)<\dots<\sm(L,j)$ and since the binary expansions of the values in the subset $\{\sm(L,i),\dots,\sm(L,j)\}$ are equal on the first $t_{ij}-1$ bits 
(this comes from the definition of $t_{ij}$), we get that for the index $z$ it holds that
\begin{equation*}
    \sm(L,z)<\trunc(t_{ij},\sm(L,i))+2^{-t_{ij}}
\end{equation*}
and 
\begin{equation*}
    \sm(L,z+1) \geq \trunc(t_{ij},\sm(L,i))+2^{-t_{ij}}.
\end{equation*}
This concludes the proof.
}

\end{proof}

}

\section*{Appendix B}
\begin{lemma}
    \label{lemma_3}
    For any non-increasing dyadic distribution $\phi=(\phi_1,\dots,\phi_m)$, let $\nu\,'= (0,\phi_1,0,\dots,0,\phi_m,0)$ be its fully extended distribution, then there does not exist any alphabetic code with lengths $L=\langle \ell_1,\dots,\ell_{2\,m+1}\rangle$ defined as follows
    \begin{equation*}
    \ell_i = \begin{cases}
            x_i & \text{ if } \nu\,'_i = 0,\\
            \lceil-\log_2 \nu\,'_i \rceil +1 & \text{ if } \nu\,'_i > 0,\\
        \end{cases}
    \end{equation*}
    where $x_i\in \mathbb{N}_+$ are arbitrarily chosen values.
\end{lemma}
\begin{proof}
    Without loss of generality let us assume that $\phi_m > 0$. Otherwise, we can simply show that the lemma holds for $(\phi_1,\dots,\phi_{m-1})$ and hence for $\phi$.
    
    Let $L = \langle \ell_1, \dots, \ell_{2\,m+1} \rangle$ be the
    list of lengths as defined in the statement of the Lemma.
    We will consider two cases. The first 
    corresponds to the case in which,
    for each $i\in\{1,3,\dots, 2\,m-1\}$, we have $ x_i= \ell_i \geq \ell_{i+1}$ and $x_{2\,m+1}= \ell_{2\,m+1} \geq \ell_{2\,m}$;
     the second case is for all other possible choices of the $x_i$.
    
    \textbf{Case 1}:
    Recall that, for $i\in \{2,\dots,m\}$, we have that $\alpha_i = \min(\ell_{i-1},\ell_i)$. Since $\phi$ is non-increasing, for each $i\in\{2,4,\dots,2\,m-2\}$, we have that $\ell_i \leq \ell_{i+2}$. Thereby, for each $i\in\{2,\dots,2\,m\}$, it  holds that $\alpha_i \leq \alpha_{i+1}$. In particular, for each $i\in \{2,4,\dots,2\,m\}$, we have that 
    {
    $$\alpha_i = \alpha_{i+1}=\lceil -\log_2 \phi_{\frac{i}{2}} \rceil +1 = -\log_2 \phi_{\frac{i}{2}} + 1.$$}
    
    Putting it all together we obtain that
    \begin{equation*}
        \begin{aligned}
            \sm(L, i) &= \trunc(\alpha_i, \sm(L, i-1))+2^{-\alpha_i}\\ 
            & = \sm(L, i-1)+2^{-\alpha_i}\\
            & = \sm(L, i-1)+2^{-\left(-\log_2 \phi_{\lfloor \frac{i}{2}\rfloor} + 1\right)}\\
            & = \sm(L, i-1) + \frac{\phi_{\lfloor \frac{i}{2}\rfloor}}{2},
        \end{aligned}
    \end{equation*}
    where the {second} equality holds {because  the $\alpha_i$ are non-decreasing, 
    and therefore the binary representation of $\sm(L,i-1)$ contains bits equal to 1 only in
    the positions less than or equal to $\alpha_i$}.
    
    Since we add up each $\phi_i$ exactly once without any truncation, $\sm(L, 2\,m+1) = \sum_{i=1}^m \phi_i = 1$ and by Theorem \ref{condition_1} there does not exist any alphabetic code for the lengths $L$.
    
    \textbf{Case 2}: In this case, there is at least one $\ell_i = x_i$ that does not satisfy the condition of Case 1. Let us suppose that $\sm(L, 2\,m+1) < 1$ and show that this leads to a contradiction.
    
    Consider the sub-case in which we increase by 1 $x_i$, with $1<i<2m+1$.
    Let $L'= \langle \ell_1,\dots,\ell_{i-1}, \ell_i +1, \ell_{i+1},\dots, \ell_{2\,m+1} \rangle$ be the 
    list of lengths after increasing such $x_i$ by 1. Let us show that $\sm(L', 2\,m+1)$ can only decrease 
    with respect to $\sm(L, 2\,m+1)$. Since $x_i$ can only affect $\alpha_i$ and $\alpha_{i+1}$, we have four cases:
    \begin{itemize}
        \item $\alpha_i$ increases by 1 and $\alpha_{i+1}$ does not change:  
        \begin{equation*}
            \begin{aligned}
                \sm(L', i) &= \trunc(\alpha_i + 1, \sm(L, i-1))+ 2^{-(\alpha_i + 1)}\\
                &\leq \trunc(\alpha_i, \sm(L, i-1)) + 2^{-(\alpha_i + 1)}\\
                &\quad+ 2^{-(\alpha_i + 1)}\\
                &= \trunc(\alpha_i, \sm(L, i-1)) + 2^{-\alpha_i}\\
                &= \sm(L, i).
            \end{aligned}
        \end{equation*}
        \item  $\alpha_i$ does not change and $\alpha_{i+1}$ increases by 1:  
        \begin{equation*}
            \begin{aligned}
                \sm(L', i+1) &= \trunc(\alpha_{i+1} + 1, \sm(L, i)) \\
                &\quad+ 2^{-(\alpha_{i+1} + 1)}\\
                &\leq \trunc(\alpha_{i+1}, \sm(L, i)) \\
                &\quad+ 2^{-(\alpha_{i+1} + 1)} + 2^{-(\alpha_{i+1} + 1)}\\
                &= \trunc(\alpha_{i+1}, \sm(L, i)) + 2^{-\alpha_{i+1}}\\
                &= \sm(L, i+1).
            \end{aligned}
        \end{equation*}
        \item Both $\alpha_i$ and $\alpha_{i+1}$ increase by 1: then $\alpha_i = \alpha_{i+1}$ and 
        \begin{align*}
                \sm(L', i+1) &= \trunc(\alpha_{i+1} + 1, \sm(L', i))\\
                &\quad+ 2^{-(\alpha_{i+1}+1)}\\
                &=\trunc(\alpha_{i+1} + 1,\\
                &\quad\quad\trunc(\alpha_i + 1, \sm(L,i-1))\\
                &\quad\quad+ 2^{-(\alpha_i+1)} )\\ 
                &\quad+ 2^{-(\alpha_{i+1}+1)} \\
                &= \trunc(\alpha_i + 1, \sm(L, i-1)) \\
                &\quad+ 2^{-(\alpha_i+1)} + 2^{-(\alpha_{i+1}+1)}\\
                &\quad{\mbox{(since $\alpha_i = \alpha_{i+1}$})}\\
                &=\trunc(\alpha_i + 1, \sm(L, i-1)) \\
                &\quad+ 2^{-(\alpha_i+1)} + 2^{-(\alpha_{i}+1)}\\
                &=\trunc(\alpha_i + 1, \sm(L, i-1)) + 2^{-\alpha_i}\\
                &\leq \trunc(\alpha_i, \sm(L, i-1))\\ 
                &\quad+ 2^{-\alpha_i} + 2^{-(\alpha_i+1)}\\
                &<\trunc(\alpha_i, \sm(L, i-1))\\
                &\quad+ 2^{-\alpha_i} + 2^{-\alpha_{i+1}} \\
                &= \trunc(\alpha_{i+1},\\&\quad\quad \trunc(\alpha_i,\sm(L,i-1))+2^{-\alpha_i})\\&\quad+ 2^{-\alpha_{i+1}}\\
                &= {\trunc(\alpha_{i+1}, \sm(L,i)) + 2^{-\alpha_{i+1}}}\\
                &= \sm(L, i+1).
        \end{align*}
        \item  Both $\alpha_i$ and $\alpha_{i+1}$ do not change: the final sum remains the same.
    \end{itemize}
Therefore, we have that $\sm(L', 2\,m+1)\leq \sm(L,2\,m+1)$.

Consider now the sub-cases in which we increase by 1 $x_1$ or $x_{2\,m+1}$. In such sub-cases, only $\alpha_2$ and $\alpha_{2\,m+1}$ can be affected. Let us consider the case when $x_1$ is increased and let $L' = \langle \ell_1+1,\dots,\ell_{2\,m+1} \rangle$ be the list of lengths after increasing $x_1$. We have two possibilities:
\begin{itemize}
    \item $\alpha_2$ increases by 1:
    \begin{equation*}
        \begin{aligned}
            \sm(L',2)&=\trunc(\alpha_{{2}} + 1, \sm(L, 1)) + 2^{-\alpha_{{2}} -1}\\
            &= \trunc(\alpha_{{2}}, \sm(L, 1)) + 2^{-\alpha_{{2}} - 1}\\
            &<\trunc(\alpha_{{2}}, \sm(L, 1)) + 2^{-\alpha_{{2}}}\\
            &=\sm(L, 2)
        \end{aligned}
    \end{equation*}
    \item  $\alpha_2$ does not change: then $\sm(L',2)=\sm(L,2)$.
\end{itemize}
A similar reasoning holds for $x_{2\,m+1}$ too. By denoting with $L' = \langle \ell_1,\dots,\ell_{2\,m+1}+1 \rangle$ the list of lengths after increasing $x_{2\,m+1}$, we have two possibilities:
\begin{itemize}
    \item $\alpha_{2\,m+1}$ increases by 1:
    \begin{equation*}
        \begin{aligned}
            \sm(L',2\,m+1) &= \trunc(\alpha_{2\,m+1}+1, \sm(L,2\,m))\\
            &\quad+2^{-(\alpha_{2\,m+1} + 1)}\\
            &\leq \trunc(\alpha_{2\,m+1}, \sm(L,2\,m))\\
            &\quad+2^{-(\alpha_{2\,m+1} + 1)}+2^{-(\alpha_{2\,m+1} + 1)}\\
            &= \trunc(\alpha_{2\,m+1}, \sm(L,2\,m))\\
            &\quad+ 2^{-\alpha_{2\,m+1}}\\
            &=\sm(L, 2\,m+1)
        \end{aligned}
    \end{equation*}

    \item  $\alpha_{2\,m+1}$ does not change: then $\sm(L', 2\,m+1) = \sm(L, 2\,m+1)$.
\end{itemize}
Therefore, in such cases too, we have that $\sm(L', 2\,m+1)\leq \sm(L, 2\,m+1)$.
\smallskip

However, the above conclusions
readily lead to a contradiction since, if we repeat the process iteratively for each $x_i$,
we obtain that when $x_i = \ell_{i+1}$ for each $i\in\{1,3,\dots,2\,m-1\}$, and $x_{2\,m+1} = \ell_{2\,m}$, it holds that $\sm(L', 2\,m+1) < 1$ and by Theorem \ref{condition_1} there exists an alphabetic codes with lengths $L'$. But, for such lengths $L'$, we are in Case 1 and we know that this is not possible. Thus, even for values of $x_i$ arbitrarily chosen, we have that $\sm(L, 2\,m+1)\geq 1$ and by Theorem \ref{condition_1} there does not exist any alphabetic code
for the list of length $L$.
\end{proof}

\end{document}